\documentclass[11pt]{article}
\usepackage{latexsym,epsfig,amssymb,amsmath,amsthm,color,url,bbm,pdfsync,ulem}
\usepackage{graphicx,float,bm}
\usepackage{subcaption}
\usepackage{amsfonts}
\usepackage{graphics}
\usepackage{mathrsfs,hyperref}
\usepackage{algorithm}
\usepackage{algpseudocode}

\usepackage{fancyhdr, graphicx,tikz}
\usetikzlibrary{arrows,positioning} 

\numberwithin{equation}{section}
\allowdisplaybreaks
\usepackage[round,comma]{natbib}

\newtheorem{lemma}{Lemma}[section]

\newtheorem{theorem}[lemma]{Theorem}
\newtheorem{proposition}[lemma]{Proposition}
\newtheorem{corollary}[lemma]{Corollary}
\newtheorem{remark}[lemma]{Remark}
\newtheorem{condition}{Condition}[section]
\newcommand{\indep}{\perp \!\!\! \perp}

\title{Graphical lasso for extremes}
\author{
Phyllis Wan\footnote{Erasmus University Rotterdam; Econometric Institute, Burg.\ Oudlaan 50, 3062 PA Rotterdam, the Netherlands;
email: wan@ese.eur.nl}
\and
Chen Zhou\footnote{Erasmus University Rotterdam; Econometric Institute, Burg.\ Oudlaan 50, 3062 PA Rotterdam, the Netherlands;
email: zhou@ese.eur.nl}}
\date{}

\begin{document}
\maketitle

\begin{abstract}
In this paper, we estimate the sparse dependence structure in the tail region of a multivariate random vector, potentially of high dimension. The tail dependence is modeled via a graphical model for extremes embedded in the H\"usler-Reiss distribution.  We propose the \textit{extreme graphical lasso} procedure to estimate the sparsity in the tail dependence, similar to the Gaussian graphical lasso in high dimensional statistics.  We prove its consistency in identifying the graph structure and estimating model parameters. The efficiency and accuracy of the proposed method are illustrated by simulations and real data examples.
\end{abstract}
{\footnotesize \noindent\it Keywords and phrases: graphical lasso; graphical models; multivariate extreme value statistics; high dimensional statistics; H\"usler-Reiss distribution. } \\
{\footnotesize {\it AMS 2010 Classification:} 62G32; 62H12; 62F12.} 


\section{Introduction}

Consider a Gaussian random vector with mean zero and covariance matrix $\Sigma$. In a Gaussian graphical model, the precision matrix $\Theta:=\Sigma^{-1}$ encodes the conditional dependence structure among the variables -- variables $i$ and $j$ are conditionally independent given the rest of the variables if and only if $\Theta_{ij}=0$, see, e.g.~\cite{Lauritzen}.  

Given an estimate of the covariance matrix $\hat\Sigma$, the {\it graphical lasso} method estimates a sparse $\Theta$ using $L_1$-regularization by 
\begin{equation} \label{eq:glasso}
	\underset\Theta{\arg\min} \left\{- \log|\Theta| + tr \left( \hat\Sigma\Theta\right) + \gamma \sum_{i\ne j}|\Theta_{ij}| \right\},
\end{equation}
where $\gamma$ is a tuning parameter for regularization;
see, e.g.~\cite{yuan2007model}, \cite{banerjee2008model} and \cite{fri2008}. The advantage of the graphical lasso method is two folds.  First, it reveals the conditional dependence among the underlying random variables by producing a sparse estimate of $\Theta$.  Second, it provides a reliable estimation of $\Theta$ and $\Sigma$ in the high dimensional case where classical covariance estimation fails. The theoretical properties of the graphical lasso procedure are investigated in \cite{rothman2008sparse} and \cite{ravikumar2011high}.

In this paper, we aim to estimate the sparse dependence structure in the {\it tail region} among high dimensional random variables. With the characterization of tail dependence, one can further conduct statistical risk assessment of extreme (co-)occurrences, such as systemic banking failures (e.g.~\citealp{zhou2010banks}) or compound environmental events (e.g.~\citealp{col1991}). Our approach is built on the framework of \cite{eng2018a}, which introduces graphical models for extremes by defining the conditional dependence in the tail distribution.

A parametric distribution family that can accommodates sparse graphical models for extremes is the H\"usler-Reiss (HR) distribution \citep{hue1989}. The class of HR distributions describes the non-trivial limiting tail distributions of Gaussian triangular arrays. Similar to Gaussian distribution, it is parametrized by bilateral relations. More specifically, a $d$-dimensional HR graphical model can be parametrized by a precision matrix $\Theta\in \mathbb{R}^{d\times d}$, such that the variables $i$ and $j$ are conditionally independent in the extremes given the rest of the variables if and only if $\Theta_{ij}=0$ \citep{eng2018a,hentschel2022statistical}.  

Unlike the Gaussian case, the precision matrix $\Theta$ in the HR model is not of full rank.  Due to this low-rank property, existing statistical inference procedures for estimating $\Theta$ in an HR model often start with conditioning on a chosen dimension being above a high threshold. In turn, one can only estimate $\Theta^{(k)} \in \mathbb{R}^{(d-1)\times(d-1)}$, the submatrix of $\Theta$ where the $k$-th row and $k$-th column are removed \citep{eng2018a}. Estimating an HR graphical model is therefore challenging when a sparse $\Theta$ is desired: a sparse estimate of $\Theta^{(k)}$ does not guarantee sparsity on the omitted $k$-th row and column. \cite{hentschel2022statistical} propose an estimation procedure for $\Theta$ using matrix completion when the sparsity structure of $\Theta$ was known. To date, the only sparse estimation for $\Theta$ without knowing the sparsity structure ex-ante is proposed by \cite{engelke2021learning}. They achieve this goal by aggregating sparse estimates of $\Theta^{(k)}$ for all $k=1,\ldots,d$ using a majority vote to decide whether or not each entry of $\Theta$ should be zero.  In other words, their estimation procedure requires estimating $d$ graphical models, which can be computationally intensive for large $d$. 

In this paper, we propose a direct estimate of $\Theta$ with a built-in option for sparse estimation via $L_1$-regularization.  We term it the {\it extreme graphical lasso} (EGLasso).  The core idea is as follows.  We show that by adding a positive constant $c$ to each entry of $\Theta$, the matrix
$$
	\Theta^* := \Theta + c\mathbf1\mathbf1^T
$$
is the inverse of a covariance matrix $\Sigma^*$ which can be estimated consistently from observations. To impose sparsity on the entries of $\Theta$, we only need to estimate $\Theta^*$ by shrinking the off-diagonal entries to $c$, which can be achieved in the optimization
$$
	\underset{\Theta^*}{\arg\min}  \left\{- \log|\Theta^*| + tr \left( \hat\Sigma^*\Theta^*\right) + \gamma\sum_{i\ne j}|\Theta^*_{ij}-c|\right\},
$$
where $\gamma>0$ is the regularization parameter.
The EGLasso thus achieves both graph structure identification and parameter estimation of $\Theta$ by solving a single penalized likelihood problem.  Compared to the EGLearn procedure of \cite{engelke2021learning}, which requires solving $d$ separate optimization problems, EGLasso is computationally more efficient, particularly for large $d$. On the other hand, EGLearn with neighborhood selection can achieve better graph recovery for very sparse graphs, as it exploits the conditional independence structure node-by-node. The two approaches offer complementary strengths.

We provide finite sample theory and asymptotic theory for the extreme graphical lasso. In particular, we show a consistent identification of the graph and accurate estimation of the non-sparse parameters in $\Theta$. We provide practical guidance in choosing the tuning parameters for applying the extreme graphical lasso method, based on theoretical motivations. We compare the performance of EGLasso with EGLearn through an extensive simulation study using data simulated from the max-stable H\"usler-Reiss distribution. The simulation results show that EGLasso achieves comparable graph recovery for moderately dense graphs, while being $3$--$5$ times faster at dimension $d=100$. Finally, we apply EGLasso to two real data examples and compare the results with EGLearn.

The remainder of the paper is structured as follows.  The background for HR graphical models and new insights regarding its parametrizations are presented in Section~\ref{sec:hr}.  The extreme graphical lasso method is introduced in Section~\ref{sec:eglasso}.  The finite sample and asymptotic theories are shown in Sections~\ref{sec:results} and \ref{sec:estimation}.  Finally, the performance of the method is illustrated in Sections~\ref{sec:simulations} and \ref{sec:application}.

\subsection{Notation}
We adopt the following notation.  Let $\mathbf0$ and $\mathbf1$ denote vectors whose elements are all 0's and all 1's respectively.  We allow them to denote vectors of different length in different contexts when there are no possibilities of confusion.  Let $\mathbf{e}_j$ denote the vector such that
$$
	\mathbf{e}_j = \underbrace{(0,\ldots,0,1,0,\ldots,0)^T}_{\text{Only the $j$-th entry is 1.}}.
$$
Given a square matrix $A$, $A\succ 0$ and $A\succeq 0$ denote the fact that $A$ is positive definite and positive semi-definite, respectively.
For the matrix norms: $\Vert\cdot\Vert_\infty$ is the element-wise $L_\infty$-norm, both for vectors and matrices; $|||\cdot|||_\infty$ is the $l_\infty$-operator norm for matrices, i.e.~the row-wise maxima of $L_1$-norms applied to each row. We note the following properties of these norms:
\begin{itemize}
	\item Both $\Vert\cdot\Vert_\infty$ and $|||\cdot|||_\infty$ are norms.
	\item For matrix $A$ and vector $v$, $\Vert Av\Vert_\infty\leq ||| A|||_\infty \Vert v\Vert_\infty$.
	\item For matrices $A$ and $B$ with compatible dimensions, $||| AB|||_\infty\leq ||| A|||_\infty||| B|||_\infty$.
\end{itemize}


\section{H\"usler-Reiss graphical models} \label{sec:hr}

In this section, we describe the class of HR graphical models.  In particular, we offer new insights into its various parametrizations, including their geometric interpretations and their connections to each other.

\subsection{Graphical models for extremes}
Consider a random vector $\mathbf{X}=(X_1,\ldots,X_d)$. Denote $\tilde X_k = \frac{1}{1-F_k(X_k)}$, where $F_k$ is the marginal distribution function of $X_k$. Then $\mathbf{\tilde X}=(\tilde X_1,\ldots,\tilde X_d)$ is a random vector with standard Pareto marginals and summarizes the dependence structure of $\mathbf{X}$. Following multivariate extreme value theory, we assume that $\mathbf{\tilde X}$ belongs to the domain of attraction of a multivariate extreme value distribution, i.e. the limit of its component-wise maxima converges to a non-degenerate distribution. Specifically, given i.i.d.~copies of  $\mathbf{\tilde X}$, $\mathbf{\tilde X}^i = (\tilde X_1^i,\ldots,\tilde X_d^i),i\in\mathbb{N}$, there exists a random vector $\mathbf{Z} = (Z_1,\ldots,Z_d)$ such that
\begin{equation} \label{eq:maxima}
	P(\mathbf{Z} \le \mathbf{z}) := \lim_{n\to\infty} P\left( \max_{i=1,\ldots,n} \tilde X_1^i \le nz_1,\ldots, \max_{i=1,\ldots,n} \tilde  X_d^i \le nz_d\right) =  G(\mathbf{z}),
\end{equation}
where each marginal distribution of $G$ is Fr\'echet-distributed. 
By writing
$$
	G(\mathbf{z}) = \exp\left(- \Lambda(\mathbf{z})\right),
$$
where $\Lambda$ is a Radon measure on the cone $\mathcal{E} = [0,\infty)^d\backslash\{\mathbf0\}$ and $\Lambda(\mathbf{z})$ is shorthand for $\Lambda([0,\infty)^d\backslash[\mathbf0,\mathbf{z}])$, the measure $\Lambda$ is known as the exponent measure and characterizes the dependence structure of $\mathbf{X}$ in the tail region, see, e.g. \citet[Chapter 6.1]{haan2006extreme}. 

The domain of attraction condition \eqref{eq:maxima} can be equivalently expressed in terms of threshold exceeding. Consider the exceedances of $\mathbf{\tilde X}$ where its $L_\infty$-norm $\|\mathbf{\tilde X}\|_\infty$ is higher than a certain threshold. There exists a random vector $\mathbf{Y}$ such that
\begin{equation} \label{eq:pot}
	P(\mathbf{Y} \le \mathbf{z}) := 	\lim_{u\to \infty} P\left(\left. \frac{\mathbf{\tilde X}}{u}\le \mathbf{z}\right| \|\mathbf{\tilde  X}\|_\infty > u\right) = \frac{\Lambda(\mathbf{z}\wedge\mathbf1) - \Lambda(\mathbf{z})}{\Lambda(\mathbf{1})}.
\end{equation}
Here the random vector $\mathbf{Y}$ is defined with support on the $L$-shaped set $\mathcal{L}=\{\mathbf{x}\in\mathcal{E}:\|\mathbf{x}\|_\infty>1\}$.  
Its distribution is known as a multivariate Pareto distribution; see \cite{rootzen2006multivariate}.


The framework of graphical models for extremes \citep{eng2018a} considers the conditional independence of the threshold exceedance limit $\mathbf{Y}$ in \eqref{eq:pot}.  Since $\mathbf{Y}$ is defined on the $L$-shaped set $\mathcal{L}=\{\mathbf{x}\in\mathcal{E}:\|\mathbf{x}\|_\infty>1\}$, which is not a product space, the notion of conditional independence is instead defined on the subspace $\mathcal{L}^k = \{\mathbf{x}\in\mathcal{L}: x_k>1\}$ for each $k$.  Define the random vector $\mathbf{Y}^k\stackrel{d}{=}\mathbf{Y}|Y_k>1$.  Then $\mathbf{Y}$ is said to exhibit conditional independence in extremes between component $i$ and $j$ if and only if $Y_i^k$ and $Y_j^k$ are conditionally independent for all $k$,
\begin{equation} \label{eq:extreme_ci}
	Y_i \indep_e Y_j | \mathbf{Y}_{\backslash\{i,j\}} \quad \Leftrightarrow \quad  \forall k\in\{1,\ldots,d\}: Y_i^k \indep Y_j^k | \mathbf{Y}^k_{\backslash\{i,j\}},
\end{equation} 
where $\mathbf{Y}_{\backslash\{i,j\}}$ ($\mathbf{Y}^k_{\backslash\{i,j\}}$) indicates all other dimensions in $\mathbf{Y}$ ($\mathbf{Y}^k$) excluding $\{i,j\}$. 

Let $\mathcal{G}=(V,E)$ be a graph defined by a set of nodes $V=\{1,\ldots,d\}$ and a set of undirected edges between pairs of nodes $E \subset V\times V$.  A graphical model for extremes based on graph $\mathcal{G}$ has a multivariate Pareto distribution $\mathbf{Y}$ that satisfies
$$
	 \{i,j\}\notin E \quad \Leftrightarrow \quad Y_i \indep_e Y_j | \mathbf{Y}_{\backslash\{i,j\}}.
$$
In the case where the exponent measure $\Lambda$ admits a density function $\lambda$, the conditional independence in extreme is equivalent to the decomposition of $\lambda$ \citep{eng2018a}:
\begin{equation} \label{eq:lambda:decomp}
	Y_i \indep_e Y_j | \mathbf{Y}_{\backslash\{i,j\}} \quad \Leftrightarrow \quad \lambda(\mathbf{y}) = \frac{\lambda_{\backslash\{i\}}(\mathbf{y}_{\backslash\{i\}})\lambda_{\backslash\{j\}}(\mathbf{y}_{\backslash\{,j\}})}{\lambda_{\backslash\{i,j\}}(\mathbf{y}_{\backslash\{i,j\}})},
\end{equation}
where $\mathbf{y}_{\backslash{A}}$ denotes the entries of $\mathbf{y}$ outside of the index set $A$.



\subsection{H\"usler-Reiss graphical models}

The class of H\"usler-Reiss (HR) distributions is the family of distributions that describes the non-trivial tail limiting distribution of Gaussian triangular arrays \citep{hue1989}.  It has been used as a canonical parametric family and a counterpart to Gaussian distribution in modelling extremes.  There exist various parametrizations of the HR distributions,  the algebraic relationships between which can be found in \cite{rottger2023}.  In the following, we motivate these parametrizations from a geometric point of view, which leads to simpler proofs of existing results and additional insights.  The proofs of the propositions in this section are presented in Appendix~\ref{appendix:prop:proofs}.

\subsubsection{Charaterization using variagoram matrix $\Gamma$ and sub-covariance matrices $\Sigma^{(k)}$'s}

In its original definition \citep{hue1989}, a $d$-dimensional HR model is parametrized by a matrix $\Gamma \in \mathbb{R}^{d\times d}$, defined on the parameter space
$$
	\mathcal{D} =  \left\{\Gamma| \Gamma^T = \Gamma,\ \text{diag}(\Gamma)=\mathbf0;\ \mathbf{x}^T\Gamma \mathbf{x} \le 0,\ \text{for any $\mathbf{x}\ne \mathbf0$ such that $\mathbf{x}^T\mathbf1=\mathbf0$})\right\}.
$$
In this paper, we only consider the full-rank HR model defined by the sub-parameter space
$$
	\mathcal{D}_0 =  \left\{\Gamma| \Gamma^T = \Gamma,\ \text{diag}(\Gamma)=\mathbf0;\ \mathbf{x}^T\Gamma \mathbf{x} < 0,\ \text{for any $\mathbf{x}\ne \mathbf0$ such that $\mathbf{x}^T\mathbf1=\mathbf0$})\right\},
$$
such that the exponent measure $\Lambda$ admits a density\footnote{The same assumption is imposed in all previous literature on extremal graphical models, such as \cite{eng2018a}, \cite{engelke2021learning}, \cite{hentschel2022statistical} and \cite{rottger2023}.}.

The following proposition provides the geometric interpretation of this parameter space. Although it is a well-known result, a proof is included for completeness.

\begin{proposition}\label{prop:variogram}
Given a multivariate random vector $\mathbf{W}=(W_1,\ldots,W_d)$, define its variogram matrix to be $\Gamma = (\Gamma_{ij}) \in \mathbb{R}^{d\times d}$ such that
$$
	\Gamma_{ij} = E(W_i-W_j)^2.
$$
The parameter space $\mathcal{D}_0$ is the collection of variogram matrices $\Gamma$ for all random vectors $\mathbf{W}$ with positive definite covariance matrices.

\end{proposition}

Given the covariance matrix $\Sigma$ of $\mathbf{W}$, its variogram matrix can be derived from
\begin{eqnarray}
	 \Gamma_{ij} &=& E[(W_i - W_j)]^2 \nonumber\\
	 &=& E(W_i - W_k)-(W_j - W_k)^2 \nonumber\\
	 &=& E(W_i - W_k)^2 + E(W_j - W_k)^2 - 2E[(W_i - W_k)(W_j-W_k)] \nonumber\\
	 &=& \Sigma_{ii} + \Sigma_{jj}  - 2\Sigma_{ij}. \label{eq:sigma:gamma}
\end{eqnarray}
However, given a variogram $\Gamma$, the choice of covariance matrix of 
the corresponding random vector $\mathbf{W}$ needs not be unique.  For example, given any univariate random variable $A$, $\mathbf{W} - A\cdot \mathbf1$ has the same variogram as $\mathbf{W}$.  Therefore the converse of Proposition~\ref{prop:variogram} fails: any $\Gamma \in \mathcal{D}_0$ can also be the variogram matrix of a $\mathbf{W}$ whose covariance matrix $\Sigma$ is not of full rank.

It is possible to impose additional assumptions on the structure of $\mathbf{W}$ to obtain a one-to-one correspondence between its covariance and variogram matrices.   One option is to restrict $W_k\equiv0$ for a given $k$.  
Assume that $\mathbf{W}$ has full-rank covariance matrix and variogram $\Gamma$.  Consider $\mathbf{W}^{(k)} := \mathbf{W} - W_k \cdot \mathbf1$.  Then $\mathbf{W}^{(k)}$ has variogram $\Gamma$ and covariance matrix $\tilde\Sigma^{(k)}$ with its elements defined by
\begin{eqnarray}
	 \tilde\Sigma^{(k)}_{ij} &:=& E[W^{(k)}_iW^{(k)}_j] \nonumber\\
	 &=& E[(W_i - W_k)(W_j - W_k)] \nonumber\\
	 &=& \frac{1}{2}\left(E(W_i - W_k)^2 + E(W_j - W_k)^2 - E[(W_i - W_j)^2\right) \nonumber\\
	 &=&\frac{1}{2}\left(\Gamma_{ik} + \Gamma_{jk} - \Gamma_{ij}\right) \label{eq:sigma_k}.
\end{eqnarray}
Combining \eqref{eq:sigma:gamma} and \eqref{eq:sigma_k}, there exists a one-to-one transformation between $\Gamma$ and $\tilde\Sigma^{(k)}$.  Consequently, we can use $\tilde\Sigma^{(k)}$, for any fixed $k$, to parametrize an HR distribution.  Note that $\tilde\Sigma^{(k)}$ is degenerate on the $k$-th row and the $k$-th column.  Let $\Sigma^{(k)}$ be the $(d-1)\times(d-1)$ matrix constructed by removing the $k$-th row and the $k$-th column from $\tilde\Sigma^{(k)}$.  For convenience, we index the rows and columns of $\Sigma^{(k)}$ using the original row and column numbers from $\tilde\Sigma^{(k)}$.  By definition, $\tilde\Sigma^{(k)}$ is of rank $d-1$, therefore $\Sigma^{(k)}$ is of full rank. 

This parametrization can be used to express the density of the exponent measure $\Lambda(\cdot)$ of the HR model \citep{engelke2015estimation}:
$$
	\lambda(\mathbf{y}) = y_k^{-2} \prod_{i\ne k} y_{i}^{-1} \phi_{d}(\mathbf{y}'_{k};\Sigma^{(k)}),
$$
where $\phi_{d}(\cdot;\Sigma^{(k)})$ is the density of a centered $d$-dimensional Gaussian distribution with covariance matrix $\Sigma^{(k)}$ and $\mathbf{y}'_k=\{\log(y_i/y_k) + \Gamma_{ik}/2\}_{i=\{1,\ldots,d\}\backslash k}$.  Note again that this formula holds for any $k=1,\ldots,d$.

\subsubsection{Characterization using precision matrix $\Theta$}

Denote the inverse of $\Sigma^{(k)}$ to be $\Theta^{(k)}:= \left(\Sigma^{(k)}\right)^{-1}$.  It follows from the decomposition of the exponent density \eqref{eq:lambda:decomp} that
$$
	Y_i \indep_e Y_j|\mathbf{Y}_{\backslash\{i,j\}} \quad \Leftrightarrow \quad \Theta^{(k)}_{ij} = 0, \quad \forall k \ne i,j.
$$
In other words, $\Theta^{(k)}$ serves as the precision matrix for the graph $\mathcal{G}$ excluding the node $k$.
The following result from \cite{eng2018a} guarantees a precision matrix $\Theta$ for the full graph $\mathcal{G}$.

\begin{proposition}[Lemma 1, \citealp{eng2018a}]
For $k\ne k'$, the following relationships hold
\begin{eqnarray}
	\Theta^{(k')}_{ij} = \Theta^{(k)}_{ij} &,& \text{if } i,j \ne k; \nonumber\\
	\Theta^{(k')}_{ik} = -\sum_{l\ne k}\Theta^{(k)}_{il} &,& \text{if } i \ne k \text{ and } j=k; \nonumber\\
	\Theta^{(k')}_{kk} = \sum_{m,l\ne k}\Theta^{(k)}_{ml} &,& \text{if } i = k \text{ and } j=k. \nonumber
\end{eqnarray}
\end{proposition}

Define matrix $\Theta \in \mathbb{R}^{d\times d}$ such that for any $k$,
\begin{eqnarray}
	\Theta_{ij} = \Theta^{(k)}_{ij} &,& \text{if } i,j \ne k; \nonumber\\
	\Theta_{ik} = -\sum_{l\ne k}\Theta^{(k)}_{il} &,& \text{if } i \ne k \text{ and } j=k; \nonumber\\
	\Theta_{kk} = \sum_{m,l\ne k}\Theta^{(k)}_{ml} &,& \text{if } i = k \text{ and } j=k. \nonumber
\end{eqnarray}
Then $\Theta$ satisfies
$$
	Y_i \indep_e Y_j|\mathbf{Y}_{\backslash\{i,j\}} \quad  \Leftrightarrow \quad \Theta^{(k)}_{ij} = 0,\, i,j \ne k \quad\Leftrightarrow \quad \Theta_{ij} = 0.
$$
Hence $\Theta$ is referred to as the {\it precision matrix} for the HR graphical model \citep{hentschel2022statistical}. 

Let $\tilde\Theta^{(k)}$ be the $(d\times d)$-matrix obtained by augmenting $\Theta^{(k)}$ with a $k$-th column and a $k$-th row of 0 entries.
Then for any $k$, $\Theta \in \mathbb{R}^{d\times d}$ can be written as
$$
	\Theta =  (I - \mathbf1 \mathbf{e}_k^T) \cdot\tilde\Theta^{(k)} \cdot (I -  \mathbf{e}_k\mathbf1^T),
$$
where recall that $\mathbf{e}_k$ is the vector with 1 on the $k$-th entry and 0 everywhere else. 
In particular, removing the $k$-th column and the $k$-th row from $\Theta$ results in $\Theta^{(k)}$.  
There exists one-to-one relationship between $\Theta$ and $\Theta^{(k)}$, and hence $\Sigma^{(k)}$, for each $k$.  Therefore $\Theta$ can be use to parametrize an HR distribution.  

\subsubsection{Characterization using covariance matrix $\Sigma$}

In the following, we define a covariance matrix which serves as the `inverse' fo $\Theta$.  Given a variogram $\Gamma$, we consider the corresponding random vector $\mathbf{W}$ with the restriction $\sum_{k=1}^d W_k \equiv 0$.
The following proposition shows that, given any variogram $\Gamma \in \mathcal{D}_0$, one can always construct a random vector whose components sum to zero and whose variogram equals $\Gamma$, and that such a vector admits a well-defined covariance matrix.

\begin{proposition} \label{prop:sigma:gamma}
Let $\mathbf{W}$ be a random vector with variogram $\Gamma \in \mathcal{D}_0$.  Then the random vector $\mathbf{W}' := \mathbf{W} - \bar{W} \cdot \mathbf1$, which satisfies $\sum_{k=1}^d W'_k = 0$, has variogram $\Gamma$ and covariance matrix
\begin{equation} \label{eq:definition of Sigma}
	\Sigma := -\frac{1}{2}\left(I - \frac{\mathbf{1}\mathbf{1}^T}{d}\right) \Gamma \left(I - \frac{\mathbf{1}\mathbf{1}^T}{d}\right).
\end{equation}
\end{proposition}

Conversely, $\Gamma$ can be obtained from $\Sigma$ via \eqref{eq:sigma:gamma}.  Therefore $\Sigma$ can be used to parametrize an HR distribution.

In the following, we provide the connections between $\Sigma$, $\tilde\Sigma^{(k)}$ and $\Theta$.

\begin{proposition}[Relationship between $\Sigma$ and $\tilde\Sigma^{(k)}$] \label{prop:sigma:sigma_k}
For any $k$, $\tilde\Sigma^{(k)}$ can be written as
$$
	\tilde\Sigma^{(k)} =  (I - \mathbf1 \mathbf{e}_k^T) \cdot \Sigma \cdot (I -  \mathbf{e}_k\mathbf1^T).
$$
Conversely, given $\tilde\Sigma^{(k)}$'s for all $k$, $\Sigma$ can be written as
\begin{equation} \label{eq:sigma:sigma_k:1}
	\Sigma = \frac{1}{d} \sum_{k=1}^d \tilde\Sigma^{(k)} + \frac{1}{d} tr(\Sigma) \cdot  \mathbf1\mathbf1^T.
\end{equation}
\end{proposition}

\begin{remark}
Given the fact that $\Sigma\mathbf1 =0$, from \eqref{eq:sigma:sigma_k:1}, 
$$
	0 = \mathbf1^T\Sigma\mathbf1 = \frac{1}{d} \sum_{k=1}^d \mathbf1^T\tilde\Sigma^{(k)}\mathbf1 + \frac{1}{d} tr(\Sigma) \cdot  \mathbf1^T\mathbf1\mathbf1^T\mathbf1 =  \frac{1}{d} \sum_{k=1}^d \mathbf1^T\tilde\Sigma^{(k)}\mathbf1 + \frac{1}{d} tr(\Sigma) \cdot  d^2.
$$
Therefore $\Sigma$ can be derived from $\tilde\Sigma^{(k)}$'s by
\begin{equation} \label{eq:sigma:sigma_k}
	\Sigma = \frac{1}{d} \sum_{k=1}^d \tilde\Sigma^{(k)} - \frac{1}{d^3} \left( \sum_{k=1}^d \mathbf1^T\tilde\Sigma^{(k)}\mathbf1\right) \cdot  \mathbf1\mathbf1^T.
\end{equation}
\end{remark}

\begin{proposition}[Relationship between $\Sigma$ and $\Theta$] \label{prop:sigma:theta}
Let $\Sigma$ have the eigendecomposition
$$
	\Sigma = \sum_{k=1}^d \lambda_k \mathbf{u}_k\mathbf{u}_k^T,
$$
where $\{\lambda_k\}$ is the set of ordered eigenvalues of $\Sigma$ and $\{\mathbf{u}_k\}$ is the corresponding set of orthonormal eigenvectors.
Then 
$$
	0 = \lambda_1 <\lambda_2\le \cdots \le \lambda_d,
$$
and 
$$
	\mathbf{u}_1 = \frac{1}{\sqrt{d}}\mathbf1.
$$
Meanwhile, $\Theta$ can be written as
$$
	\Theta = \sum_{k=2}^d \frac{1}{\lambda_k} \mathbf{u}_k\mathbf{u}_k^T.
$$
In particular, for any $M>0$,
$$
	 \left( \Sigma + \frac{M}{d}\mathbf{1}\mathbf{1}^T\right)^{-1} =\left( \Sigma + M\mathbf{u}_1\mathbf{u}_1^T\right)^{-1} =\Theta + \frac{1}{M}\mathbf{u}_1\mathbf{u}_1^T =\Theta + \frac{1}{dM}\mathbf{1}\mathbf{1}^T.
$$
\end{proposition}

\begin{remark}
\cite{hentschel2022statistical} study the matrix $\Sigma$ and show that it is the Moore-Penrose inverse of $\Theta$ such that
$$
	\lim_{M\to\infty} \left( \Sigma + M\mathbf{1}\mathbf{1}^T\right)^{-1} = \Theta.
$$
Proposition~\ref{prop:sigma:theta} provides a stronger result with a simpler proof.
\end{remark}

\begin{corollary}[Parameter space of $\Sigma$ and $\Theta$] \label{prop:parameter:space}
Let
$$
	\mathcal{L} = \{A \in \mathbb{R}^{d\times d}: A \succeq 0, A \mathbf1=\mathbf0, \text{rank}(A) = d-1\}.
$$
Then any $\Sigma \in \mathcal{L}$ or $\Theta \in \mathcal{L}$ characterizes an HR distribution.
\end{corollary}

\section{The extreme graphical lasso} \label{sec:eglasso}

In this section, we propose a sparse estimator for the precision matrix $\Theta$.  As we have seen in Section~\ref{sec:hr}, the HR model can be equivalently parametrized by $\Gamma$, $\Sigma^{(k)}$'s, $\Sigma$, $\Theta^{(k)}$'s or $\Theta$.  Among these parameter matrices, $\Sigma^{(k)}$'s can be most conveniently estimated from data.  Given an HR multivariate Pareto distribution $\mathbf{Y}$, \cite{engelke2015estimation} show that
\begin{equation} \label{eq:cond:guassian}
	\left(\log\mathbf{Y}_{-k} - \log(Y_k)\cdot \mathbf1\right) |_{Y_k>1} \sim N(\Gamma_{-k, k}/2, \Sigma^{(k)}).
\end{equation}
Therefore $\Sigma^{(k)}$'s can be estimated from the sample covariance matrix of the thresholded transformed observations.  We will describe the details of the estimator in Section~\ref{sec:hat_sigma_k}.  For now, we consider using estimators $\hat\Sigma^{(k)}$'s to derive an estimator for $\Theta$.  Since each $\hat\Sigma^{(k)}$ requires different truncation and therefore makes use of different information from the data, it is preferred to utilize the estimators $\hat\Sigma^{(k)}$ for all $k=1,\ldots,d$.  

A straightforward non-sparse estimator for $\Theta^{(k)}$ is $\left(\hat\Sigma^{(k)}\right)^{-1}$, which can be written as
\begin{equation} \label{eq:hat_theta_k}
	\hat\Theta^{(k)} := \left(\hat\Sigma^{(k)}\right)^{-1} = \underset{\Theta^{(k)}\succ0}{\arg\min} \left\{- \log|\Theta^{(k)}| +tr\left(\hat{\Sigma}^{(k)}\Theta^{(k)}\right) \right\}.
\end{equation}
The right hand side of \eqref{eq:hat_theta_k} corresponds to a pseudo negative log-likelihood in \eqref{eq:cond:guassian}, which allows us to interpret $\hat\Theta^{(k)}$ as a pseudo-MLE estimator (`pseudo' as to acknowledge the fact that the mean of the Gaussian distribution is estimated separately using the sample mean and not jointly with the covariance matrix).  It can also be viewed as the minimization of a Bregman divergence between the observed and estimated inverse covariance matrix (see, e.g.~\citealp{ravikumar2011high}) and therefore can be interpreted outside the Gaussian context.

Given estimators $\hat\Sigma^{(k)}$ for all $k$, we may aggregate \eqref{eq:hat_theta_k} for all $k$ and solve for $\Theta$ to minimize
\begin{equation} \label{eq:opt2}
	\underset{\Theta\in \mathcal{L}}{\arg\min} \sum_{k=1}^d \left\{- \log|\Theta^{(k)}| +tr\left(\hat{\Sigma}^{(k)}\Theta^{(k)}\right) \right\},
\end{equation}
where each $\Theta^{(k)}$ denote the submatrix of $\Theta$ with the $k$-th row and column removed.

Let $S$ be the estimator of $\Sigma$ from \eqref{eq:sigma:sigma_k} defined as follows:
\begin{equation} \label{eq:sigma:hat}
	S:= \frac{1}{d} \sum_{k=1}^d \hat{\tilde\Sigma}^{(k)} -  \left(\frac{1}{d^3} \sum_{k=1}^d \mathbf1^T\hat{\tilde\Sigma}^{(k)} \mathbf1\right)  \mathbf1\mathbf1^T,
\end{equation}
where $\hat{\tilde\Sigma}^{(k)}$ denote the estimator for $\tilde\Sigma^{(k)}$ augmented from $\hat\Sigma^{(k)}$.  
The following proposition shows that the aggregated negative log-likelihood \eqref{eq:opt2} can be written as a negative log-likelihood-like expression of $\Theta$ using $S$.

\begin{proposition} \label{prop:mle}
The following problems have equivalent solutions:
\begin{equation} \label{eq:theta:mle}
	\underset{\Theta\in\mathcal{L}}{\arg\min} \sum_{k=1}^d \left\{- \log|\Theta^{(k)}| +tr\left(\hat{\Sigma}^{(k)}\Theta^{(k)}\right) \right\} = \underset{\Theta\in\mathcal{L}}{\arg\min} \left\{-\log|\Theta|_+ + tr(S \Theta)\right\}.
\end{equation}
Here $|\cdot|_+$ denote the pseudo-determinant (product of all nonzero eigenvalues) of a matrix.
\end{proposition}

\begin{remark}
Lemma~5.1 of \cite{rottger2023} motivates the same equivalence as in \eqref{eq:theta:mle}, where $S$ is calculated from a different procedure.   They propose to first obtain estimate $\hat{\Gamma}_k$ from $\hat\Sigma^{(k)}$ via \eqref{eq:sigma:gamma}, average all estimates by $\hat\Gamma = \frac{1}{k}\sum_{k=1}^d \hat{\Gamma}_k$, then obtain the estimate $S$ from the aggregated estimator $\hat\Gamma$ via \eqref{eq:definition of Sigma}.   The two formulation results in the same estimator $S$.

\end{remark}

\subsection{Graphical lasso estimator of $\Theta$}

Following Proposition~\ref{prop:sigma:theta}, write $\Theta$ in its eigen decomposition form
$$
	\Theta = \sum_{k=2}^d \frac{1}{\lambda_k} \mathbf{u}_k\mathbf{u}_k^T.
$$
Then
$$
	|\Theta|_+ = \frac{1}{\prod_{k=2}^d \lambda_k}.
$$
For any $M>0$,
$$
	\left|\Theta + \frac{1}{dM} \mathbf1\mathbf1^T\right| = \left|\Theta + \frac{1}{M}  \mathbf{u}_1\mathbf{u}_1^T\right| = \frac{1}{M} \cdot \frac{1}{\prod_{k=2}^d \lambda_k} = \frac{1}{M}  \cdot |\Theta|_+.
$$
Combining with the fact that $\Theta \cdot \mathbf1=\mathbf0$ for any $\Theta \in \mathcal{L}$, the estimator for $\Theta$ in \eqref{eq:theta:mle} can also be written as
$$
	\underset{\Theta\in\mathcal{L}}{\arg\min} \left\{-\log\left|\Theta + \frac{1}{dM} \mathbf1\mathbf1^T\right| + tr\left(\left(S + \frac{M}{d}\mathbf1\mathbf1^T\right) \left(\Theta+\frac{1}{dM} \mathbf1\mathbf1^T\right)\right)\right\},
$$
which is similar to minimizing the Gaussian negative log-likehood function of $\Theta^* := \Theta + \frac{1}{dM} \mathbf1\mathbf1^T$ using covariance estimate $S^* := S + \frac{M}{d}\mathbf1\mathbf1^T$.
To impose zero-entries to $\Theta$, we impose $L_1$-penalties on the absolute entries of $\Theta$:
$$
	\underset{\Theta\in\mathcal{L}}{\arg\min} \left\{-\log\left|\Theta + \frac{1}{dM} \mathbf1\mathbf1^T\right| +tr\left(\left(S + \frac{M}{d}\mathbf1\mathbf1^T\right) \left(\Theta+\frac{1}{dM} \mathbf1\mathbf1^T\right)\right) + \gamma \sum_{i\ne j}|\Theta_{ij}|\right\}.
$$
Unlike from the classical graphical lasso, as $\gamma$ increases, this will not result in a sparse graph.  Intuitively, due to the constraint that $\sum_{i=1}^d\Theta_{ij} = 0$ for any $j$, $\Theta_{ij}$ will approach $-\Theta_{jj}/(d-1)$ as $\gamma$ increases, not $0$.  See, for example, Theorem~3.1 of \cite{ying2020}. 

To obtain a sparse estimator, we propose to relax the search domain of $\Theta$ and instead solve the following problem:
\begin{equation} \label{eq:eglasso:optimization}
	\hat{\Theta}_{lasso}:=\underset{\Theta + \frac{1}{dM} \mathbf1\mathbf1 \succ 0} {\arg\min}\left\{-\log\left|\Theta + \frac{1}{dM} \mathbf1\mathbf1^T\right| + tr\left(\left(S + \frac{M}{d}\mathbf1\mathbf1^T\right) \left(\Theta+\frac{1}{dM} \mathbf1\mathbf1^T\right)\right) + \gamma \sum_{i\ne j}|\Theta_{ij}|\right\}.
\end{equation}
We term the solution of \eqref{eq:eglasso:optimization} the {\it extreme graphical lasso estimator} of $\Theta$.

In Section~\ref{sec:results}, we show that under suitable conditions, the extreme graphical lasso is consistent in estimating $\Theta$ and recovering underlying sparse graph structure, despite the relaxation of the search domain.

\subsection{Solving the extreme graphical lasso}

Denote $c:= \frac{1}{dM}$, $\Theta^* := \Theta + \frac{1}{dM} \mathbf1\mathbf1^T$ and $S^* := S + \frac{M}{d}\mathbf1\mathbf1^T$.  Then the problem in \eqref{eq:eglasso:optimization} can be reformulated as
$$
	\underset{\Theta^* \succ 0}{\arg\min} \left\{-\log\left|\Theta^*\right| + tr\left(S^*\Theta^*\right) + \gamma \sum_{i\ne j}|\Theta^*_{ij}-c|\right\}.
$$
This is similar to the Gaussian graphical lasso problem in \eqref{eq:glasso} except that we penalize the off-diagonal entries of $\Theta^*$ to a positive constant $c$ instead of $0$.  Despite the similarity, this problem cannot be solved directly by the state-of-the-art GLASSO algorithm proposed in \cite{fri2008}.  Instead, we modify another block coordinate descent algorithm, the P-GLASSO algorithm, from \cite{mazumder2012graphical}.  Similar to GLASSO, P-GLASSO is computationally efficient and is guaranteed to converge on a trajectory of positive definite matrices.  In Appendix~\ref{app:algo}, we provide the algorithm for the extreme graphical lasso and its convergence properties.

\section{Theoretical results} \label{sec:results}

In this section, we establish the non-asymptotic and asymptotic theories for the extreme graphical lasso. The goal is to learn the graphical structure for extremes and estimate the non-zero entries of $\Theta$ simultaneously. 




Recall that $(V,E)$ denotes the set of nodes and edges of the graph. Denote the maximum degree of all nodes as $D=\max_{1\leq i \leq d}\sum_{j=1}^d1_{(i,j)\in E}$. If the dimension $d=d_n\to\infty$ as the sample size $n\to\infty$, we can potentially have $D\to\infty$ and $|E|\to\infty$. Nevertheless, $D=O(d)$ and $|E|=O(d^2)$ always hold.

We first list the conditions required for learning the graph structure. The first condition concerns the structure of the graph reflected in the true covariance matrix $\Sigma$.

\begin{condition}[Mutual incoherence] \label{con:mi}  
	Given $M>0$, define $\Omega=\Sigma^*\otimes \Sigma^*$ where $\Sigma^*=\Sigma + \frac{M}{d}\mathbf{1}\mathbf{1}^T$ and $\otimes$ is the Kroneker product. We assume that there exists $0<\alpha<1$ such that
	$$	|||\Omega_{E^cE}(\Omega_{EE})^{-1}|||_{\infty}<1-\alpha,$$
	where $\Omega_{EE}\in\mathbb{R}^{|E|\times|E|}$ is the submatrix $\left(\Omega_{(i,j),(k,l)}\right)_{(i,j)\in E, (k,l) \in E}$ and $\Omega_{E^cE}$ is defined similarly.
\end{condition}

Condition~\ref{con:mi}, also referred to as the irrepresentatbility condition, is comparable with Assumption 1 in \cite{ravikumar2011high}. Such a condition is needed for theory regarding lasso-type penalization algorithms. 

The second condition concerns the tuning parameter $\gamma=\gamma_n$. In order to identify the exact graph structure, the tuning parameter should be neither too high nor too low. A low $\gamma_n$ will result in non-edges not being penalized to zero while a high $\gamma_n$ will penalize true edges to zero. A suitable $\gamma_n$ is related to the estimation error when estimating $\Sigma$. In Section~\ref{sec:estimation}, we will discuss the (non-)asymptotic properties of the estimator $S$ defined in \eqref{eq:sigma:hat} by handling the estimation error $\|S -\Sigma\|$. In this section, we formulate the bounds for $\gamma_n$ using the estimation error $\delta_n$ based on the event $\{\|S -\Sigma\|_\infty\leq \delta_n\}$.
\begin{condition} \label{con:gamma}  
    Assume that there exists $0<\epsilon<1$, such that
    $$ \|S -\Sigma\|_\infty \leq \delta_n \leq \frac{(1-\epsilon)\epsilon\alpha^2}{D|||\Sigma^*|||_\infty |||\Omega_{EE} ^{-1}|||_\infty\left[(1-\epsilon)\alpha+|||\Sigma^*|||_\infty^2|||\Omega_{EE} ^{-1}|||_\infty\right]}.$$
Further, assume that the tuning parameter $\gamma_n$ satisfies $\underline{C}_\gamma(\delta_n) \le \gamma_n \le \overline{C}_\gamma$, where
	\begin{eqnarray}
		\overline{C}_\gamma&:=&\frac{(1-\epsilon)\alpha(1-\alpha)}{D|||\Sigma^*|||_\infty |||(\Omega_{EE}) ^{-1}|||_\infty\left[(1-\epsilon)\alpha+|||\Sigma^*|||_\infty^2|||(\Omega_{EE})^{-1}|||_\infty\right]},\label{eq:upperbound rho}\\
		\underline{C}_\gamma(\delta_n)&:=&\frac{1-\alpha}{\epsilon\alpha} \cdot \delta_n. \label{eq:lowerbound rho}
	\end{eqnarray}
\end{condition}
Note that the upper bound of $\delta_n$ ensures $\underline{C}_\gamma \le \overline{C}_\gamma$ and thus the existence of a suitable $\gamma_n$. In Section~\ref{sec:estimation}, we show that our estimator $S$ satisfies $\delta_n\stackrel{p}{\to} 0 $ as $n\to\infty$. Therefore, this upper bound assumption is satisfied for sufficiently large $n$ with high probability.



The following theorem provides the concentration bounds for $\hat\Theta_{lasso}$ for fixed $n$. The proof is given in Appendix \ref{appendix: proof for Section 4}.
\begin{theorem} \label{thm:main}
    Assume that Conditions \ref{con:mi} and \ref{con:gamma} hold. 
    Denote 
    \begin{equation}\label{eq:definition r}
    r_n:=\frac{|||\left(\Omega_{EE} \right)^{-1}|||_\infty}{1-\alpha} \cdot \gamma_n.
    \end{equation}
    Then on the event $ \{\|S -\Sigma\|_\infty\leq \delta_n\}$, 
    the extreme graphical lasso algorithm specified in \eqref{eq:eglasso:optimization} has a unique solution $\hat\Theta_{lasso}$. Denote the estimated edges as $\hat E:=\{(i,j):\hat\Theta_{lasso,ij} \neq 0\}$. Then
	$$
		\hat E\subseteq E
	$$
	and $$\Vert\hat\Theta_{lasso}-\Theta\Vert_\infty\leq  r_n.$$
    In particular, if $\min\{|\Theta_{ij}|; (i,j) \in E, i\neq j\}>r_n$, then $\hat E=E$.
\end{theorem}
\begin{remark}
The condition $\min\{|\Theta_{ij}|; (i,j) \in E, i\neq j\}>r_n$ is a minimal signal strength requirement: true edges can only be recovered if the corresponding entries in $\Theta$ are sufficiently large relative to the estimation error $r_n$. This condition is standard in the graphical lasso literature; see, e.g., Theorem~1 in \cite{ravikumar2011high}. In Theorem~\ref{thm:asymptotic} below, this condition is absorbed into the requirement that $\delta_n^{-1}\min\{|\Theta_{ij}|; (i,j) \in E\}\to\infty$.
\end{remark}

Next, we present the asymptotic theory when $n\to\infty$. The following asymptotic result follows immediately from Theorem \ref{thm:main}. 


\begin{theorem} \label{thm:asymptotic}
    Assume that Condition \ref{con:mi} holds. Assume that there exists a sequence $\delta_n\to 0$ as $n\to\infty$, such that the covariance matrix estimate $S$ satisfies
    $$P(\|S -\Sigma\|_\infty>\delta_n)\to 0,$$
   and
	$$\frac{1}{\delta_n}\min\{\|\Theta_{ij}\|; (i,j) \in E, i\neq j\}\to \infty.$$
    Choose the tuning parameter $\gamma_n=\underline{C}_\gamma(\delta_n)$  as in \eqref{eq:lowerbound rho}.
    Then, as $n\to\infty$, with probability tending to 1, the extreme graphical lasso estimator \eqref{eq:eglasso:optimization} has a unique solution $\hat\Theta_{lasso}$.  In addition, 
	$$P(\hat E=E)\to 1 \text{\ \ and\ \ } \Vert\hat\Theta_{lasso}-\Theta\Vert_\infty=O_P\left(\delta_n\right).$$
\end{theorem}


\section{Applying the extreme graphical lasso} \label{sec:estimation}

\subsection{Estimators for $\Sigma^{(k)}$} \label{sec:hat_sigma_k}

The standard statistical inference for the HR model relies on the following result.  Let $\mathbf{Y}$ be the multivariate Pareto distribution from an HR model with variogram $\Gamma$.  \cite{engelke2015estimation} show that
\begin{equation} \label{eq:k:gaussian}
	\left(\log\mathbf{Y}_{-k} - \log(Y_k)\cdot \mathbf1\right) |_{Y_k>1} \sim N(\Gamma_{-k, k}/2, \Sigma^{(k)}),
\end{equation}
where $\Gamma_{-k, k}$ is the $k$-th column of $\Gamma$ with the $k$-th element removed.
Given i.i.d.~observations $\mathbf{X}^i = (X_1^i,\ldots,X_d^i), 1\leq i\leq n$, drawn from $\mathbf{X}$, an empirical counterpart of $\left(\log\mathbf{Y}_{-k} - \log(Y_k)\cdot \mathbf1\right) |_{Y_k>1}$ can be constructed as follows. 
Define the transformed observations
$$
	\hat X^{i}_k= \frac{1}{1-\hat F_k(X^i_k)},
$$
where $\hat F_k(x)=\frac{1}{n+1}\sum_{i=1}^n \mathbb{I}\{X^{i}_k\leq x\}$ is the empirical distribution function based on $X^{i}_k$'s and $\mathbb{I}$ is the indicator function. Then $\mathbf{\hat X}^{i}=(\hat X^{i}_1,\ldots,\hat X^{i}_d)$ resembles a sample of $\mathbf{\tilde X}=(\tilde X_1,\ldots,\tilde X_d)$ with $\tilde X_k = \frac{1}{1-F_k(X_k)}$, albeit not i.i.d.

Consider an intermediate sequence $k_n$ such that $k_n\to \infty$ and $k_n/n\to 0$ as $n\to\infty$. Then as $n\to\infty$, $\|\mathbf{\hat X}^{i}\|_\infty>\frac{n}{k_n}$ mimicks the condition $\|\mathbf{\tilde X}\|_\infty>u$ with $u\to\infty$.
Therefore, 
$$\frac{k_n}{n}\mathbf{\hat X}^{i}\left|\|\mathbf{\hat X}^{i}\|_\infty>\frac{n}{k_n}\right.$$ 
approximately follows the same distribution as $\mathbf{Y}$.

Let $I=\{i:\|\mathbf{\hat X}^{i}\|_\infty>\frac{n}{k_n} \}=:\{i_1,\ldots,i_m\}$, where $m=|I|$, indicating the index set corresponding to $\|\mathbf{\hat X}^{i}\|_\infty>\frac{n}{k_n}$. Denote $\hat{\mathbf{Y}}_j =(\hat{Y}_{i1},\ldots, \hat{Y}_{id}):=\frac{k_n}{n}\mathbf{\hat X}^{i_j}$ for all $j=1,\ldots,m$. Let
$$
	\hat{\mathbf{W}}^{(k)}_i = \log\hat{\mathbf{Y}}_{i,-k} - \log(\hat{Y}_{ik})\cdot \mathbf1,
$$
and $I_k$ be the index set that $I_k = \{i:\hat{Y}_{ik}>1\}$ for each dimension $k=1,\ldots,d$ such that $|I_k|=k_n$.
Then $\Sigma^{(k)}$ can be estimated by
\begin{equation} \label{eq:hat:Sigma:k}
	\hat{\Sigma}^{(k)} := \frac{1}{k_n}\sum_{i \in I_k} \left(\hat{\mathbf{W}}^{(k)}_i  - \frac{1}{k_n}\sum_{i \in I_k} \hat{\mathbf{W}}^{(k)}_i \right)\left(\hat{\mathbf{W}}^{(k)}_i  - \frac{1}{k_n}\sum_{i \in I_k} \hat{\mathbf{W}}^{(k)}_i \right)^T,
\end{equation}
which is the sample covariance matrix using $\hat{\mathbf{W}}^{(k)}_i$ conditional on $\hat{Y}_{ik}>1$.  

Theoretically an estimate of $\Theta$ can be constructed via $\hat\Theta^{(k)}=\left(\hat{\Sigma}^{(k)}\right)^{-1}$.  To achieve sparsity in  $\Theta^{(k)}$, any sparse inverse covariance matrix estimation technique can be applied here.  However, reconstruction of $\Theta$ from a sparse $\hat\Theta^{(k)}$ does not guarantee sparsity on the omitted $k$-th row and column.  \cite{engelke2021learning} propose to estimate a sparse $\hat\Theta^{(k)}$ for each $k$ and then to use a majority vote to decide whether or not each entry of $\Theta$ should be zero.  This approach is shown to be effective in recovering the sparse structure of $\Theta$, when the number of dimension is at a moderate level. For high dimensional case, tuning $d$ graphical lasso models can be cumbersome.


\subsection{Convergence of $S$ to $\Sigma$}

Recall $S$ as an estimator for $\Sigma$ in \eqref{eq:sigma:hat}. We first present the assumptions needed for bounding the estimation errors for $S$. The assumptions are in line with those needed in Theorem 3 in \cite{engelke2021learning}. 

The following condition is needed regarding the tail behavior of $\mathbf{\tilde X}=(\tilde X_1, \ldots, \tilde X_d)$.

\begin{condition}[Assumption 3, \citealp{engelke2021learning}] \label{assumption:second order} Assume that all marginal distributions of the original random vector $\mathbf X$, $F_1,\ldots,F_d$, are continuous. In addition, there exists $\xi'>0$ and $K'<\infty$ independent of $d$, such that for all $J\subset\{1,\ldots,d\}$ with $|J|\in \{2,3\}$ and $q\in(0,1]$,
$$\sup_{x\in[0,1]^3}\left|\frac{1}{q}P\left(\bigcap_{i\in J}\left\{\tilde X_i>\frac{1}{qx_i}\right\}\right)-\frac{P\left\{\bigcap_{i\in J} Y_i>\frac{1}{x_i}\right\}}{P(Y_1>1)}\right|\leq K' q^{\xi'}.$$
\end{condition}
Condition \ref{assumption:second order} is a standard second order condition quantifying the speed of converegence of the tail distribution of $\mathbf{\tilde X}$ towards the limiting distribution $\mathbf{Y}$ on bounded sets. It has been imposed in other asymptotic theories in multivariate extreme value statistics, see e.g. \cite{einmahl2012m} and \cite{engelke2022structure}.

Next, we assume that the variogram in the HR distribution $\Gamma$ has bounded entries.
\begin{condition}[Bounded entries] \label{assumption:non-degenerate} Assume that the variogram $\Gamma$ satisfies that $0<\underline{\lambda}<\inf_{i\neq j}\sqrt{\Gamma_{ij}}\leq \sup_{i\neq j}\sqrt{\Gamma_{ij}}<\overline{\lambda}$, with $\underline{\lambda}$ and $\overline{\lambda}$ independent of $d$. 
\end{condition}
Condition \ref{assumption:non-degenerate} implies the boundedness in the density of the exponent measure, see Assumption 2 in \cite{engelke2021learning}: this condition is required for establishing concentration bounds for estimators of $\Gamma$. In addition, this condition implies that for any pair $(i,j)$ with $i\neq j$, $X_i$ and $X_j$ are asymptotically dependent.

Then we have the following proposition. Its proof is postponed to Appendix \ref{appendix:proof of S estimation}.
\begin{proposition}\label{prop:accuracy of S estimation}
	Assume that Conditions \ref{assumption:second order} and \ref{assumption:non-degenerate} hold. Then for any $\xi<\xi'$, there exists positive constants $C_1$, $C_2$ and $C_3$, depending on $\xi$, $\xi'$, $\underline{\lambda}$ and $\overline{\lambda}$, independent of $d$, such that for any $\varepsilon\geq \varepsilon_n:=C_2 d^3\exp\{-\frac{C_3k_n}{(\log n)^8}\}$,
	\begin{equation}\label{eq:inequality for sigma}
		P\left(\|S -\Sigma\|_\infty>\delta_n\right)\leq \varepsilon,
	\end{equation}
	where $$\delta_n = \delta_n(\varepsilon):=C_1\left\{\left(\frac{k_n}{n}\right)^{\xi}\left(\log\left(\frac{k_n}{n}\right)\right)^2+\frac{1+\sqrt{\frac{1}{C_3}\log (C_2 d^3/\varepsilon)}}{\sqrt{k_n}}\right\}.$$

	In addition, assume that $(\log n)^4\sqrt{\frac{\log d}{k_n}}\to 0$ as $n\to\infty$.  Then
	$$\|S -\Sigma\|_\infty=O_{\text{P}}\left(\left(\frac{k_n}{n}\right)^{\xi}\left(\log\left(\frac{k_n}{n}\right)\right)^2+\sqrt{\frac{\log d}{k_n}}\right), \quad n \to\infty.$$
\end{proposition}
We remark that this proposition does not require a fixed $d$ and allows for $d=d_n\to\infty$ as $n\to\infty$. Nevertheless, the condition $(\log n)^4\sqrt{\frac{\log d}{k_n}}\to 0$ as $n\to\infty$ provides an upper bound for the diverging speed of $d_n$ towards infinity. It depends not only on $n$ but also on the intermediate sequence $k_n$. Combining Theorem \ref{thm:asymptotic} and Proposition \ref{prop:accuracy of S estimation}, we immediately obtain the speed convergence of the extreme graphical lasso estimator $\hat{\Theta}_{lasso}$.

\subsection{Choice of $M$: a theoretical guidance} \label{sec:Mchoice}

In this section, we provide reasoning for choosing $M$ in practice.  We recommend to choose $M$ from the interval $M \in [\lambda_2,\lambda_d],$
where $\lambda_2$ and $\lambda_d$ are the smallest and the largest positive eigenvalues of $\Sigma$. Recall from Proposition~\ref{prop:sigma:theta} that
$$
	\Sigma = \sum_{k=1}^d \lambda_k \mathbf{u}_k \mathbf{u}_k^T,
$$
with $0=\lambda_1 <\lambda_2\leq \ldots\leq\lambda_d$ and $\mathbf{u}_1 = \frac{1}{\sqrt{d}}\mathbf1$. 

The motivation comes from the following upper bound for the mutual incoherence condition.
\begin{proposition} \label{prop:mi:upper:bound}
	The quantity $|||\Omega_{E^cE}(\Omega_{EE})^{-1}|||_{\infty}$ in Condition~\ref{con:mi} can be bounded above by
	$$
		|||\Omega_{E^cE}(\Omega_{EE})^{-1}|||_{\infty} \le \left( |E| \cdot \frac{\lambda_{\max}^2(\Omega) - \lambda_{\min}^2(\Omega)}{\lambda_{\min}^2(\Omega)}\right)^{1/2},
	$$
where $\lambda_{\max}(\Omega)$ denotes the largest eigenvalue of $\Omega$,
$$
    \lambda_{\max}(\Omega) = \max\{\lambda_d^2,M^2\},
$$
and $\lambda_{\min}(\Omega)$ denotes the smallest eigenvalue of $\Omega$,
$$
    \lambda_{\min}(\Omega) = \min\{\lambda_2^2,M^2\}.
$$
\end{proposition}

From Proposition~\ref{prop:mi:upper:bound}, choosing $M$ smaller than $\lambda_2$ or larger than $\lambda_d$ will increase the upper bound for $|||\Omega_{E^cE}(\Omega_{EE})^{-1}|||_{\infty}$.  Therefore we recommend the choice of $M \in [\lambda_2,\lambda_d]$ to minimize this upper bound.

\subsection{Choosing $M$: numerical examples}
We now investigate two examples of $d=4$ and show how the choice of $M$ affects the mutual incoherence condition $|||\Omega_{E^cE}(\Omega_{EE})^{-1}|||_{\infty}$.   Figure~\ref{fig:mi} shows the graph structures for the two examples, the star graph and the diamond graph. The Mutual Incoherence conditions for these two graphs in the classical graphical lasso setting are studied in \cite{ravikumar2011high}.

\tikzset{
    >=stealth',
    punkt/.style={
           rectangle,
           rounded corners,
           draw=black, very thick,
           text width=6.5em,
           minimum height=2em,
           text centered},
    pil/.style={
           ->,
           thick,
           shorten <=2pt,
           shorten >=2pt,}
}
\newsavebox{\mytikzpic}
\begin{lrbox}{\mytikzpic} 
     \begin{tikzpicture}
    \begin{scope}[xshift=-2cm,yshift=1cm,scale=.8]
       \node[draw,circle] (s1) at (0,-2) {$X_1$};
      \node[draw,circle] (s2) at (-2,0) {$X_2$};
      \node[draw,circle] (s3) at (2,0) {$X_3$};
      \node[draw,circle] (s4) at (0,-4.5) {$X_4$};
      \draw[-] (s1) to node [above right] {} (s3);
      \draw[-] (s1) to node [above left] {} (s2);
      \draw[-] (s4) to node [below] {} (s1);
    \end{scope}
    
     \begin{scope}[xshift=6cm,yshift=1cm,scale=.8]
       \node[draw,circle] (s1) at (0,0) {$X_2$};
      \node[draw,circle] (s2) at (-2,-2) {$X_1$};
      \node[draw,circle] (s3) at (2,-2) {$X_4$};
      \node[draw,circle] (s4) at (0,-4) {$X_3$};
      \draw[-] (s1) to node [above right] {} (s3);
      \draw[-] (s1) to node [above left] {} (s2);
      \draw[-] (s4) to node [below left] {} (s2);
      \draw[-] (s4) to node [below left] {} (s3);
      \draw[-] (s4) to node [below] {} (s1);
    \end{scope}

     \node at (-2,-3.5) {(a)};
     \node at (6,-3.5) {(b)};
  \end{tikzpicture}    

\end{lrbox}
  \begin{figure}[h]
    \centering 
    \usebox{\mytikzpic} 
    \caption{(a) Star graph; (b) Diamond graph.}
    \label{fig:mi}
\end{figure} 

\subsubsection{Star graph} \label{sub:star}
The condition $\Theta\mathbf1=\mathbf0$ imposes extra constraints on the precision matrix $\Theta$. We consider the following parameterization
$$
	\Theta = 
	\begin{pmatrix}
	x+2 & -x & -1 & -1 \\
	-x & 2x & 0 & 0 \\
	-1 & 0 & 1 & 0 \\
	-1 & 0 & 0 & 1 \\
	\end{pmatrix}
	, \quad \text{for some } x>0
$$
which reflects the star graph in Figure~\ref{fig:mi}(a).  The top two panels of Figure~\ref{fig:M} show the values of $|||\Omega_{E^cE}(\Omega_{EE})^{-1}|||_{\infty}$ plotted against the values of $M$ for star graphs with parameters $x=0.8$ and $x=1.2$.  As $[\lambda_2,\lambda_d]$ is suggested as the practical range for $M$ from Section~\ref{sec:Mchoice}, the values of $\lambda_2$ and $\lambda_d$ are indicated as dash vertical lines in each graph.  We observe that in both examples, $\lambda_2$ is a good indication of the lowest possible value of $|||\Omega_{E^cE}(\Omega_{EE})^{-1}|||_{\infty}$ while the majority of the interval $M \in [\lambda_2,\lambda_d]$ satisfies the Mutual Incoherence condition.

\subsubsection{Diamond graph}

Now consider the diamond graph in Figure~\ref{fig:mi}(b) corresponding to the precision matrix 
$$
	\Theta =  
	\begin{pmatrix}
	x+1 & -x & -1 & 0 \\
	-x & x+2 & -1 & -1 \\
	-1 & -1 & 3 & -1\\
	0 & -1 & -1 & 2 \\
	\end{pmatrix}
	, \quad \text{for some } x>0.
$$
We plot the values of $|||\Omega_{E^cE}(\Omega_{EE})^{-1}|||_{\infty}$ against the values of $M$ for parameters $x=0.8$ and $x=1.2$ on the lower two panels of Figure~\ref{fig:M}.  Again we observe that while $\lambda_2$ is close the minimum of $|||\Omega_{E^cE}(\Omega_{EE})^{-1}|||_{\infty}$, the interval $M \in [\lambda_2,\lambda_d]$ in both examples satisfies the Mutual Incoherence condition.

\begin{figure}[H]
	\centering
  	\includegraphics[width=15cm]{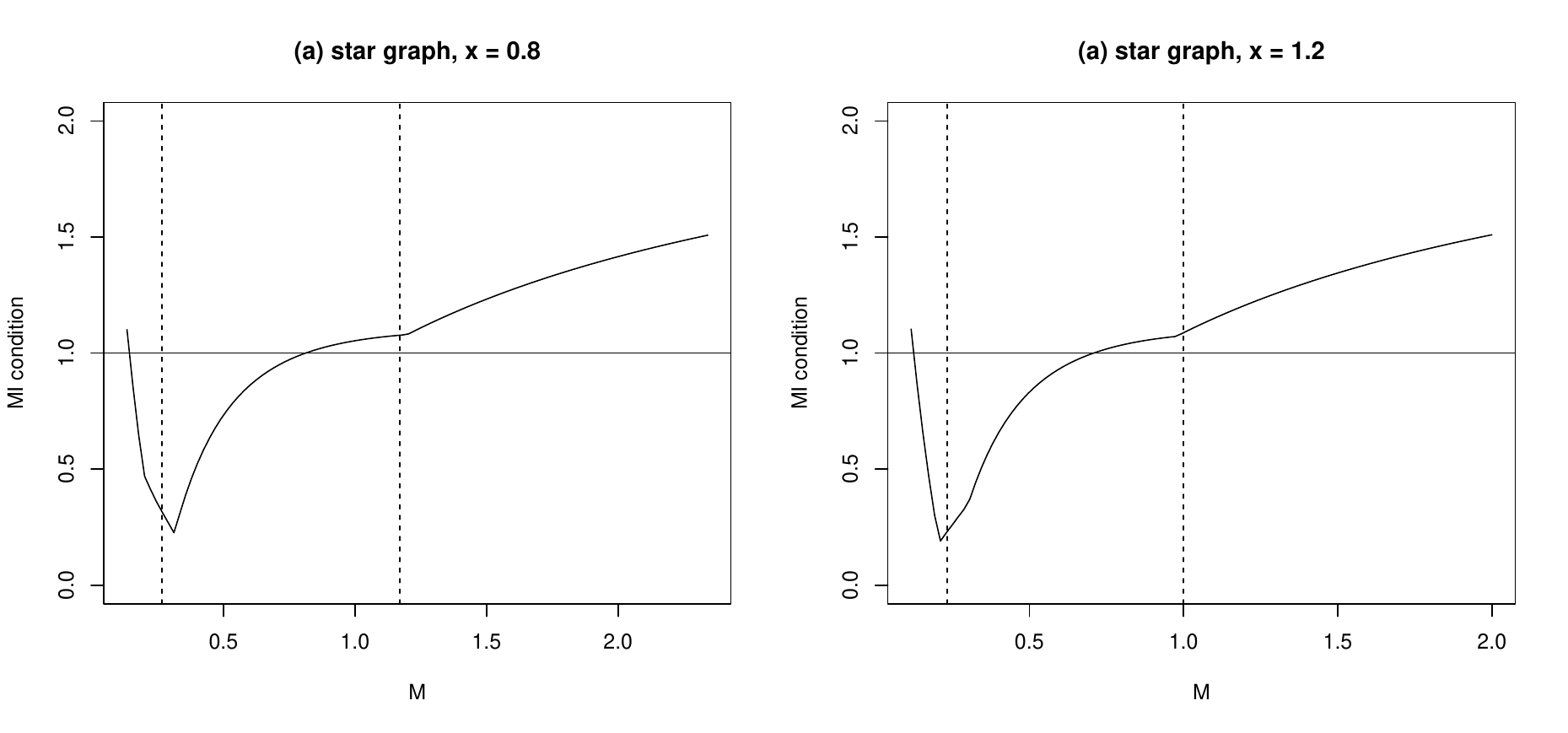}
  	\includegraphics[width=15cm]{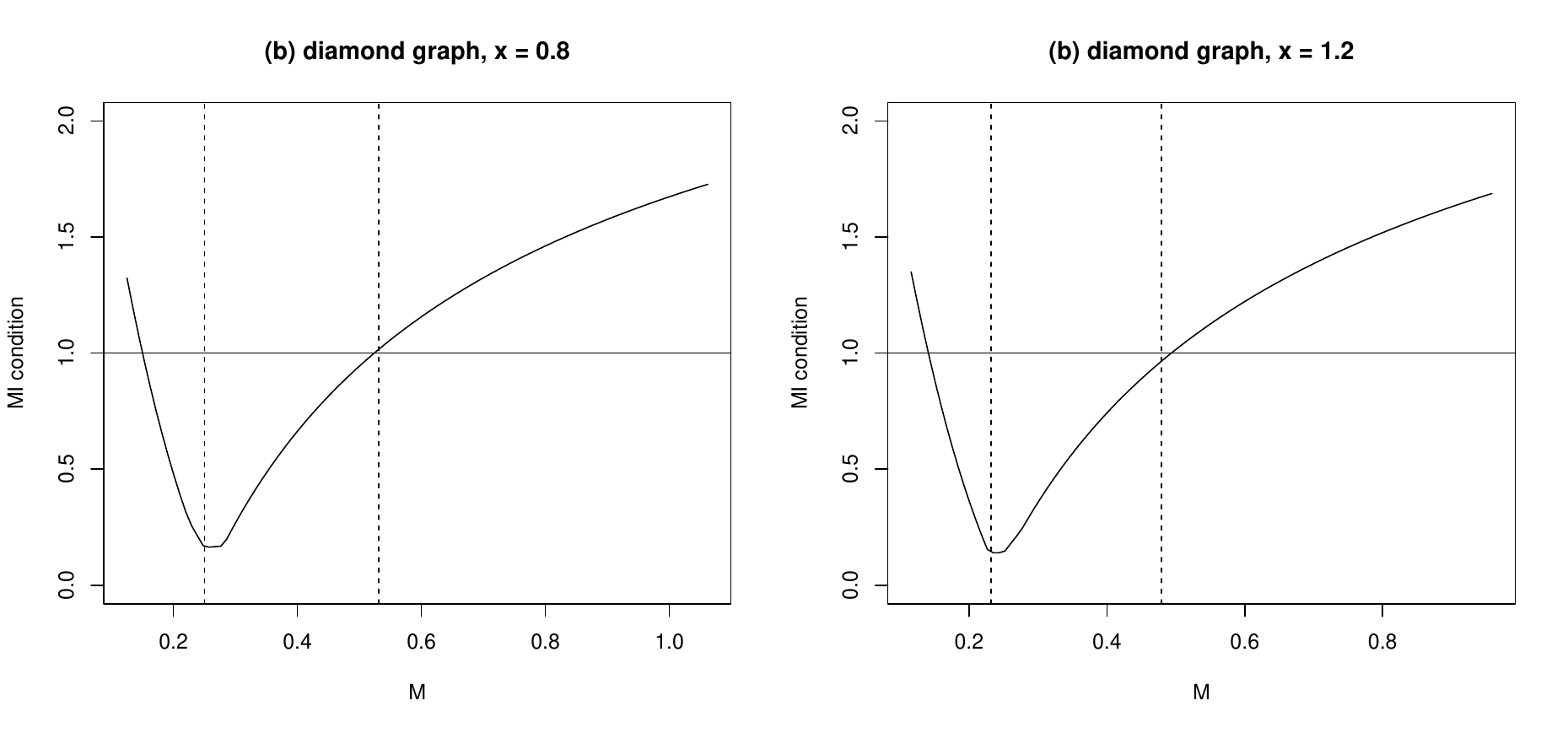}
	\caption{The curves of $|||\Omega_{E^cE}(\Omega_{EE})^{-1}|||_{\infty}$ versus $M$ for the star graph and the diamond graph, each with parameter values $x=0.8$ and $x=1.2$.  The two dash vertical lines in each graph indicates the values of $\lambda_2$ and $\lambda_d$ for the corresponding $\Sigma$. }
	\label{fig:M}
\end{figure}

\subsection{Correlation normalization}
In practice, we normalize the estimated covariance matrix $S$ to a correlation matrix before applying the extreme graphical lasso. That is, we replace $S$ by $D^{-1/2} S D^{-1/2}$, where $D = \mathrm{diag}(S)$. This normalization ensures that the penalty parameter $\gamma_n$ has a consistent scale across different dimensions and different graph structures, simplifying the tuning procedure. After estimation, the precision matrix $\hat\Theta$ is transformed back to the original scale.

\section{Simulations} \label{sec:simulations}
We conduct a simulation study to evaluate the performance of the extreme graphical lasso (EGLasso) algorithm and compare it with the EGLearn algorithm with neighborhood selection \citep{engelke2021learning}. Among the variants of EGLearn proposed in \cite{engelke2021learning}, the neighborhood selection variant outperforms graphical lasso as the base learner in every simulation scenario and has the strongest theoretical guarantees for consistent graph recovery.

\subsection{Simulation setup}
Throughout the simulation study, we sample random graph structures using the Barab\'{a}si--Albert preferential attachment model \citep{albert2002statistical} with $d$ nodes and degree parameter $q=1,2$, denoted $BA(d, q)$. The case $q=1$ corresponds to a tree graph, while $q=2$ yields a denser graph.

For a given graph structure, we sample the precision matrix $\Theta$ as follows. For each edge in the graph, we draw each non-zero upper triangular entry independently from a uniform distribution $U[-5,-2]$. The entries of $\Theta$ corresponding to non-edges are set to zero, the lower triangular entries mirror their upper triangular counterparts, and the diagonal entries are filled such that each row sum is zero. The resulting variogram matrix $\Gamma$ is then computed from $\Theta$.

We simulate $n$ observations from the max-stable H\"{u}sler--Reiss distribution with variogram $\Gamma$. After simulating $n$ max-stable observations, we apply the peaks-over-threshold transformation: we retain the $k_n$ largest observations in terms of the radial component and transform them to approximate multivariate Pareto observations. Then, we compute the sample covariance matrix $S$ from the $k_n$ transformed exceedances and apply both the EGLasso and EGLearn algorithms. For the EGLasso, the tuning parameter $M$ is chosen to be $M=\hat\lambda_2$, where $\hat\lambda_2$ is the smallest positive eigenvalue of $S$, motivated by the theoretical guidance in Section~\ref{sec:Mchoice}.

To evaluate the estimated graph, we use the F1-score defined as follows:
$$\text{F1-Score}=\frac{|E\bigcap \hat E|}{\frac{1}{2}\left(|E|+|\hat E|\right)},$$
where $E$ and $\hat E$ are the edge sets for the true and estimated graphs, respectively. For fair comparison, we report the oracle-tuned F1 score for both methods: for EGLasso, we select the penalty $\gamma_n$ that maximizes the F1 score over the grid $\log_{10}(\gamma_n) \in \{-1.2, -1.1, \ldots, 0\}$; for EGLearn, we analogously select the penalty $\rho$ that maximizes the F1 score over the grid $\rho \in\{0.175, 0.2, 0.225, \ldots, 0.475\}$. Note for both algorithms, there are 13 values in the tuning grid. This isolates the effect of the estimation method from the tuning procedure. For each setting, we repeat the simulations for $m=100$ samples. 

We maintain the sample size as $n = \lceil k_n^{1/0.7} \rceil$ so that $k_n \approx n^{0.7}$, consistent with the theoretical guidance for the intermediate sequence. We consider dimensions $d \in \{20, 100\}$ with several $k_n/d$ ratios, resulting in a range of settings from the classical regime ($k_n/d = 5$) to the high-dimensional regime ($k_n/d \leq 1$).

\subsection{Graph recovery} \label{sec:sim:f1}

Figure~\ref{fig:f1_maxstable} displays boxplots of the F1 scores for both EGLasso and EGLearn across different settings. Both methods improve as the ratio $k_n/d$ increases, reflecting the benefit of having more exceedances relative to the dimension. For the denser $BA(d,2)$ graphs, EGLasso and EGLearn achieve comparable F1 scores: at $d=100$ and $k_n/d=5$, the median F1 scores are $0.87$ and $0.93$, respectively. For the tree graphs $BA(d,1)$, EGLearn with neighborhood selection has a clear advantage, particularly at $d=100$, where it achieves near-perfect recovery (median F1 $\approx 1.0$ at $k_n/d=2.5$) while EGLasso plateaus around $0.75$. This gap can be attributed to the fact that neighborhood selection exploits the conditional independence structure node-by-node, which is particularly effective for sparse tree graphs. In the high-dimensional regime $k_n/d \leq 1$, both methods struggle to achieve a high F1 score.  Simulations based on multivariate Pareto samples (rather than max-stable samples) yield similar patterns.\footnote{Results are available from the authors upon request.}
\begin{figure}[htp]
	\centering
	\includegraphics[width=\textwidth]{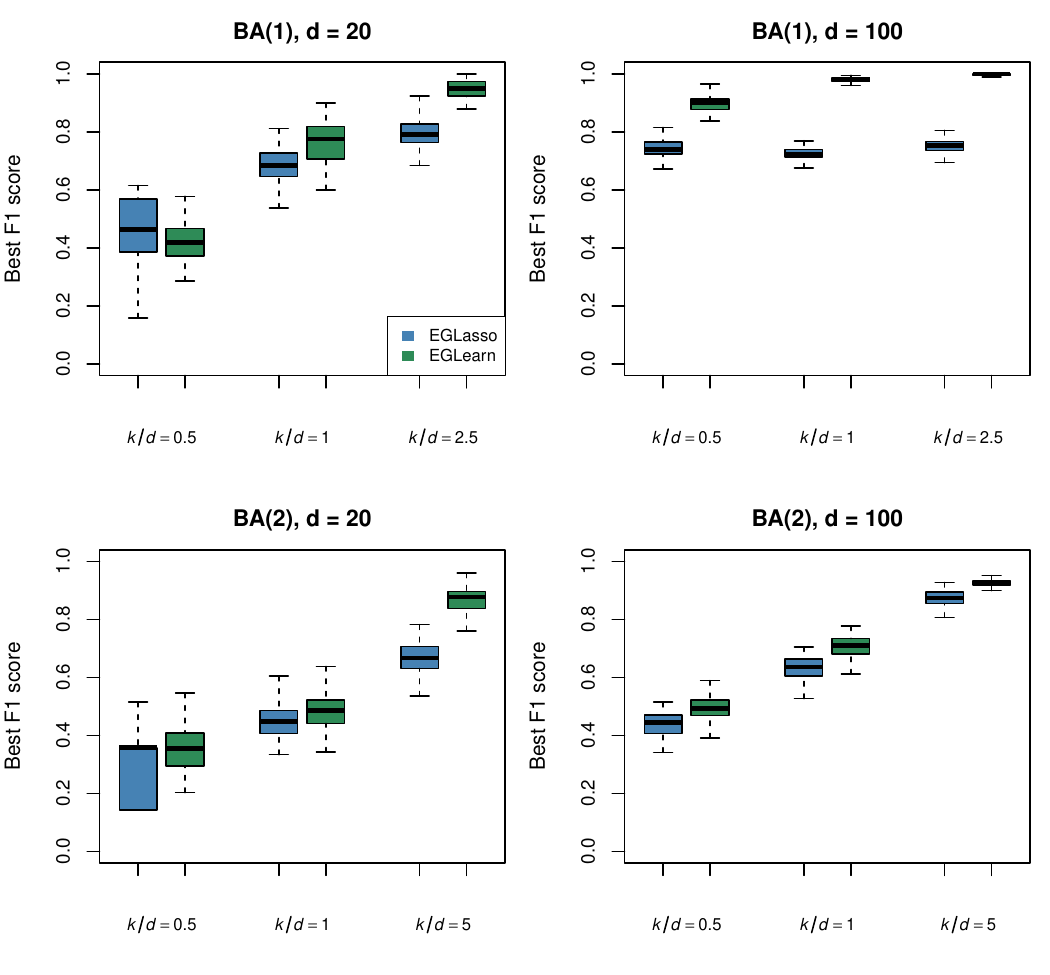}
	\caption{Oracle-tuned F1 scores for EGLasso (blue) and EGLearn (green) under the max-stable data-generating process. Upper panel: $BA(d,1)$; lower panel: $BA(d,2)$. Left: $d=20$; right: $d=100$. Boxplots are over $100$ replications. The $x$-axis shows the ratio $k_n/d$.}
	\label{fig:f1_maxstable}
\end{figure}

\subsection{Graph recovery: large sample illustration} \label{sec:sim:largesample}

To visualize the graph recovery of the EGLasso algorithm, we conduct a large-sample simulation with $d=20$ and $n=5000$, yielding $k_n = \lfloor n^{0.7} \rfloor = 388$ exceedances ($k_n/d \approx 19$). For each of $m=100$ replications, we select the oracle penalty $\gamma_n$ that maximizes the F1 score.

Figure~\ref{fig:aggregated} displays the true graph alongside the aggregated estimated graph, where the thickness of each edge reflects the fraction of replications in which it is detected. True edges are shown in black and false positives in red.  For the $BA(20,1)$ tree (upper panels), the EGLasso algorithm recovers the true edges consistently (median F1 $= 0.86$ at median $\gamma_n = 0.50$), with only a few false positive edges appearing at low frequency. For the denser $BA(20,2)$ graph (lower panels), the true edges are again detected with high frequency, though more false positives appear (median F1 $= 0.80$ at median $\gamma_n = 0.10$). The lower oracle $\gamma_n$ for $BA(20,2)$ reflects the denser true graph requiring less penalization to retain the larger number of true edges.
\begin{figure}[htp]
	\centering
	\includegraphics[width=\textwidth]{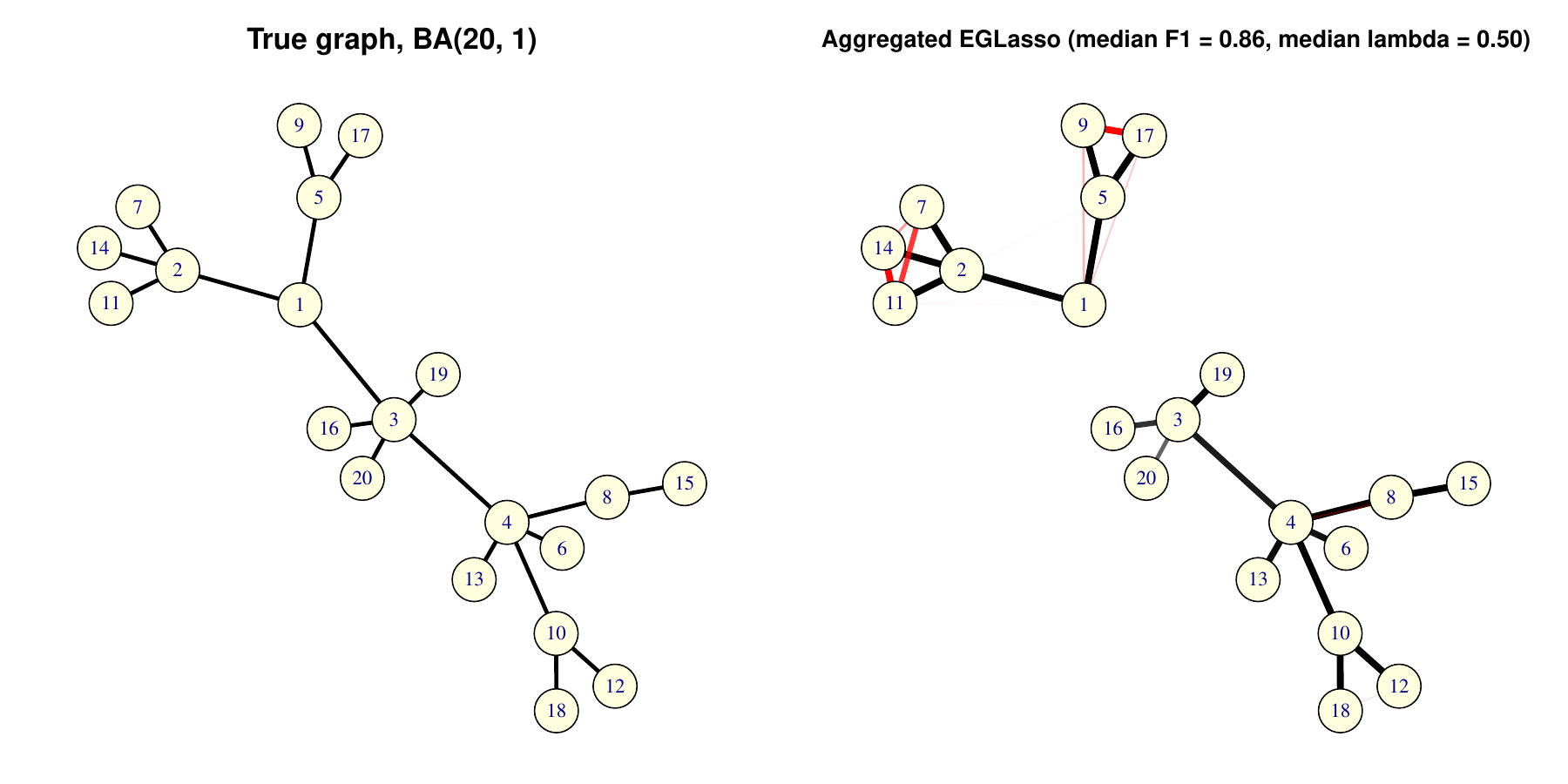}\\[0.5cm]
	\includegraphics[width=\textwidth]{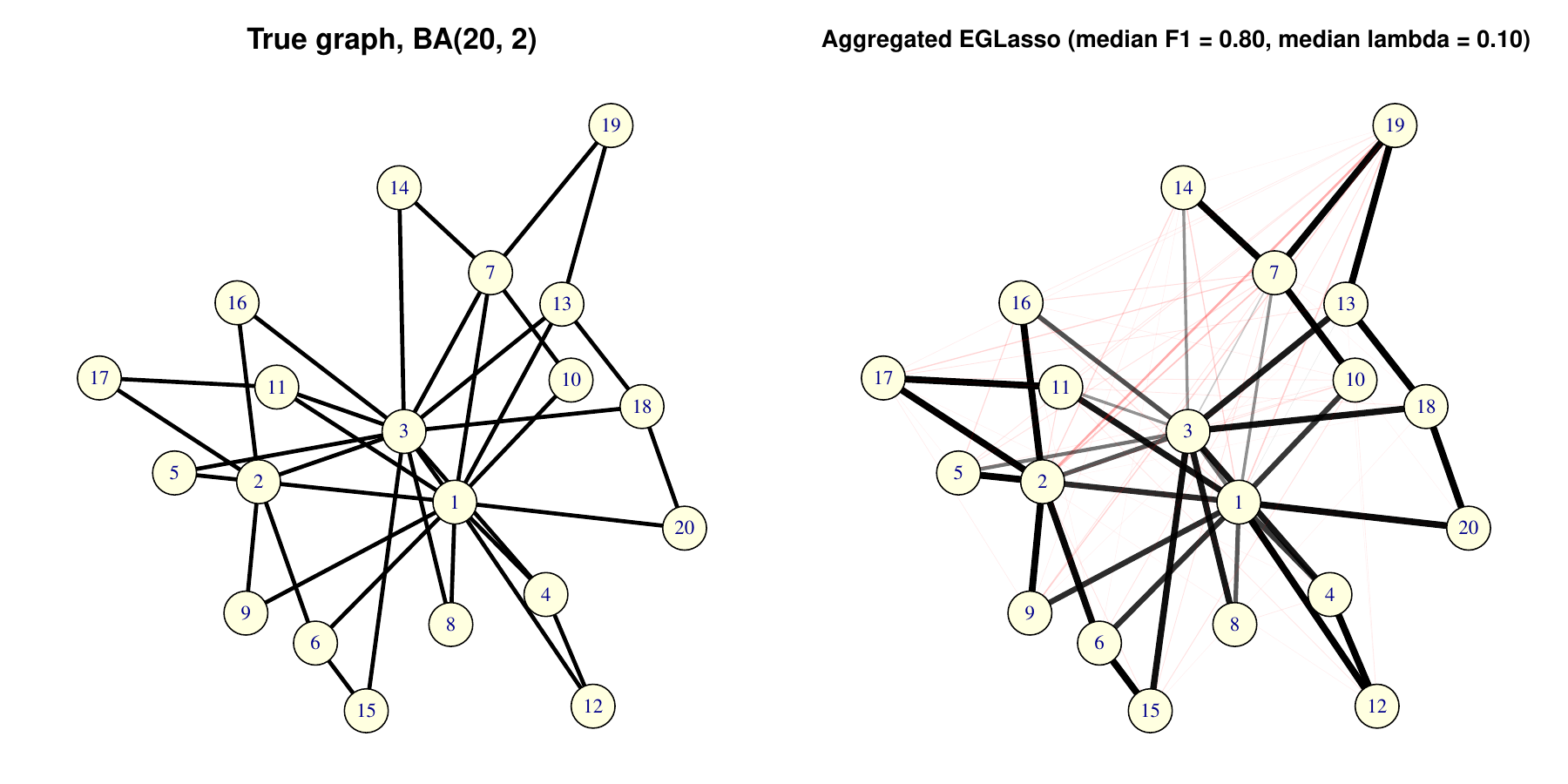}
	\caption{True graph (left) and aggregated estimated graph (right) based on 100 replications with $d=20$, $n=5000$. Upper panel: $BA(20,1)$ (tree); lower panel: $BA(20,2)$ (dense). Edge thickness in the aggregated graph reflects the detection frequency across replications. Black edges correspond to true edges; red edges are false positives. The oracle penalty $\gamma_n$ is selected per replication to maximize the F1 score.}
	\label{fig:aggregated}
\end{figure}

\subsection{Computational efficiency} \label{sec:sim:timing}

Figure~\ref{fig:timing_maxstable} displays boxplots of computation times (in seconds) for both methods. At $d=20$, both methods are fast (under one second per sample) and timing differences are negligible. At $d=100$, the computational advantage of EGLasso becomes apparent: for the $BA(100,2)$ graph, EGLasso is approximately $3$--$5$ times faster than EGLearn at $k_n/d \geq 1$, with median computation times of $5$--$9$ seconds compared to $23$--$31$ seconds for EGLearn.  For the $BA(100,1)$ tree, EGLasso achieves a $2$--$4$ times speedup at $k_n/d \geq 1$. At $k_n/d = 0.5$, the rank deficiency of the sample covariance matrix $S$ leads to more iterations in the EGLasso algorithm, reducing the speed advantage.

The computational gain of EGLasso stems from its $O(d^2)$ coordinate descent updates per iteration, compared to the $O(d^3)$ matrix inversion steps required in each of the $d$ neighborhood regressions of EGLearn. Beyond computation time, EGLearn also has substantially higher memory requirements, as it stores $d$ separate regression problems simultaneously. To provide a fair comparison, all simulations were run on a computing cluster with 185GB RAM, ensuring that EGLearn was not limited by memory constraints. On a standard workstation with 64GB RAM, the speedup of EGLasso over EGLearn increases to approximately $17$ times for the $BA(100,2)$ setting at $k_n/d=1$. The gap in computation time will be further increased if multiple replications are run in parallel, because EGLearn processes will compete for limited memory.

\begin{figure}[htp]
	\centering
	\includegraphics[width=\textwidth]{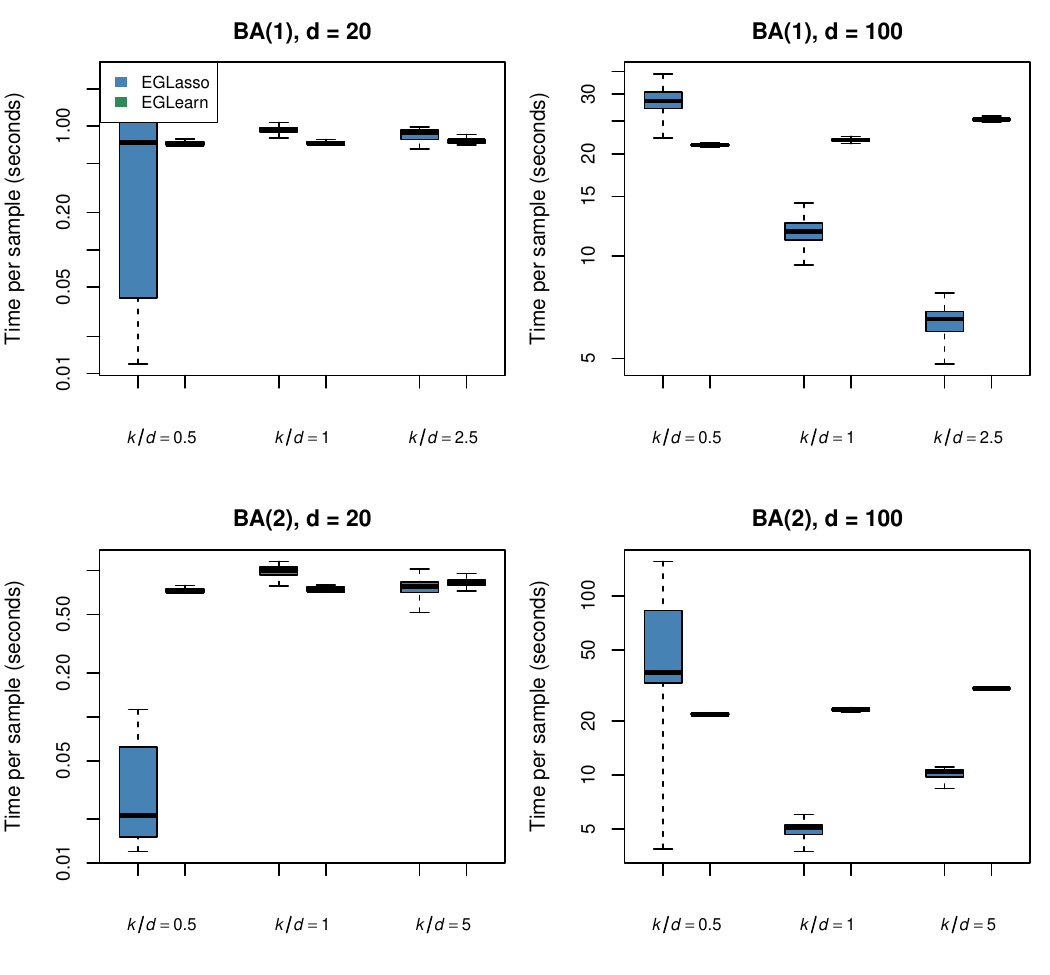}
	\caption{Computation time (seconds) for EGLasso (blue) and EGLearn (green) under the max-stable data-generating process. Upper panel: $BA(d,1)$; lower panel: $BA(d,2)$. Left: $d=20$; right: $d=100$. Boxplots are over $100$ replications. The $x$-axis shows the ratio $k_n/d$.}
	\label{fig:timing_maxstable}
\end{figure}

\section{Applications} \label{sec:application}
To illustrate the application of the extreme graphical lasso, we apply it to two real data examples. Besides applying the extreme graphical lasso, we also apply the EGLearn algorithm with neighborhood selection \citep{engelke2021learning}. For each algorithm, we present estimated graphs at three penalty levels, with the penalty parameters chosen to yield a comparable number of edges. For both algorithms, we use $k_n=\lfloor n^{0.7}\rfloor$ for estimating the covariance matrix $\Sigma$ and the variogram $\Gamma$. 

The first dataset consists of spot foreign exchange rates denominated by the British Pound
sterling between 2005 and 2020. The dataset consists of $n = 3790$ daily observations with $d = 26$ currencies, preprocessed by taking the absolute values of the de-garched log-returns, see \cite{engelke2022structure}. Figure~\ref{fig:exchange application} displays the estimated graphs for the exchange rate data. At $\gamma_n=0.35$ (62 edges), the EGLasso graph reveals a dense network with two main clusters: one among Southeast Asian currencies (HKG, TWN, MYS, THA, SGP, KOR) and another among European currencies (DNK, CZE, HUN, EUR, SWE, CHE). As the penalty increases to $\gamma_n=0.45$ (43 edges) and $\gamma_n=0.55$ (25 edges), these clusters become more pronounced and separated. The EGLearn algorithm at comparable edge counts ($\rho=0.25, 0.35, 0.55$) identifies a similar overall structure but tends to maintain connectivity as much as possible; for instance, at $\rho=0.55$ (27 edges), the EGLearn graph remains largely connected through bridge nodes, while EGLasso at $\gamma_n=0.55$ (25 edges) separates the Asian and European clusters more clearly.
\begin{figure}[htp]
	\centering
  	\includegraphics[width=15cm]{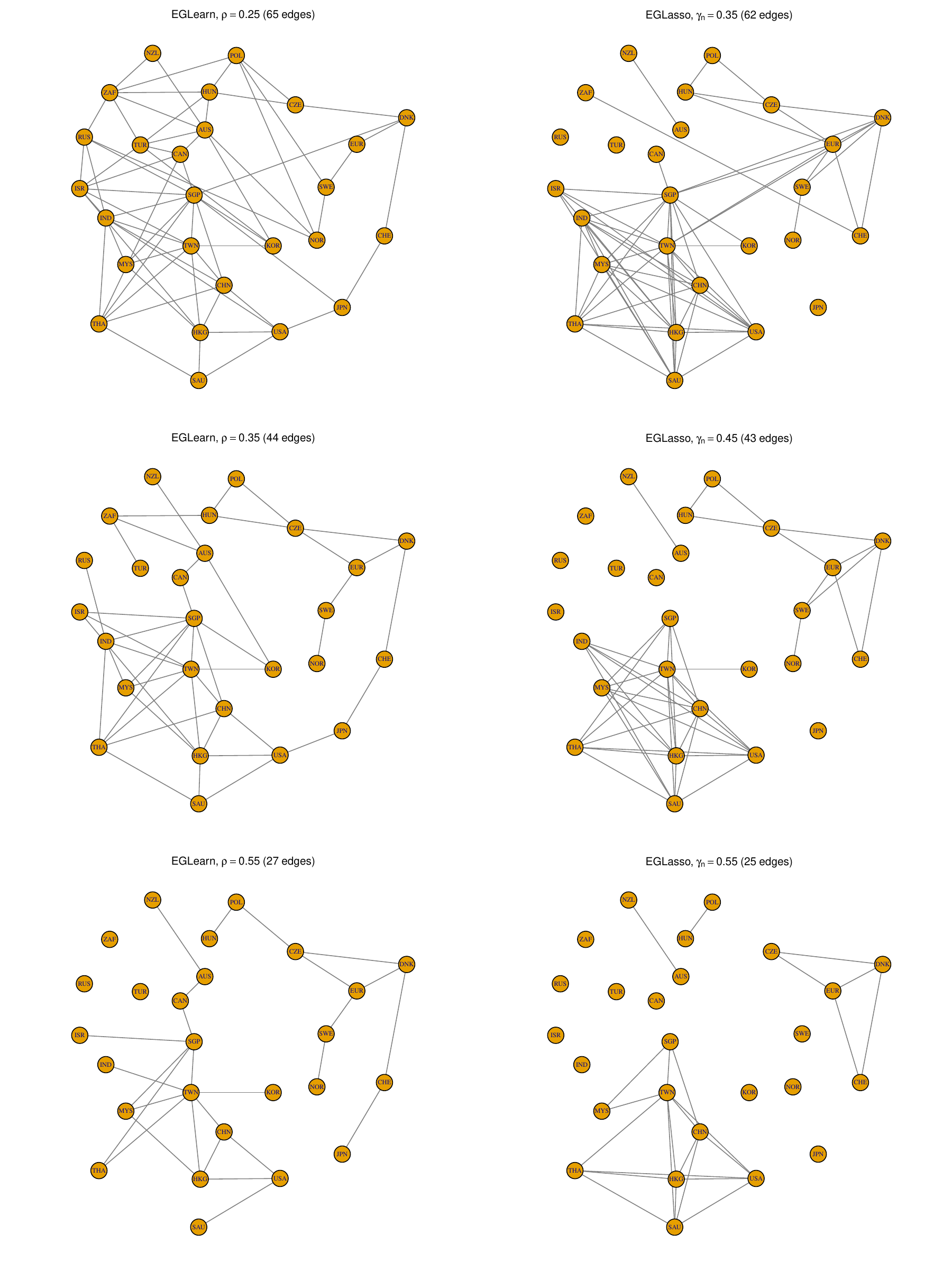}
	\caption{Exchange rate data: estimated graphical structures using EGLearn (left) and EGLasso (right) at three penalty levels. The penalties are chosen to yield a comparable number of edges.}
	\label{fig:exchange application}
\end{figure}

The second dataset concerns the river discharge at gauging stations from the upper Danube basin.  This dataset, as preprocessed in \cite{asadi2015extremes}, consists of $d=31$ stations and $n=428$ observations. Figure~\ref{fig:danube application} shows the estimated graphs at three penalty levels. At $\gamma_n=0.45$ (88 edges) and $\gamma_n=0.55$ (68 edges), the EGLasso algorithm identifies clusters of stations that reflect local hydrological catchments, with several disconnected components. The EGLearn algorithm at comparable edge counts ($\rho=0.04$ and $0.07$) produces more connected, tree-like structures. At the sparsest level ($\gamma_n=0.70$, 53 edges vs $\rho=0.11$, 51 edges), both methods reveal a similar backbone structure following the river network.

Across both applications, a structural difference between the two methods is apparent: EGLearn with neighborhood selection tends to preserve global connectivity. By contrast, EGLasso penalizes the full precision matrix and can zero out all edges of a node simultaneously, producing disconnected components even at moderate sparsity levels.  As a consequence, when the true underlying graph is expected to be connected (as in the river network), EGLearn captures the global topology more faithfully. On the other hand, when the goal is to discover local subgroups or clusters (as in the exchange rate data), EGLasso may yiled more informative subgroup structures.
\begin{figure}[htp]
	\centering
  	\includegraphics[width=15cm]{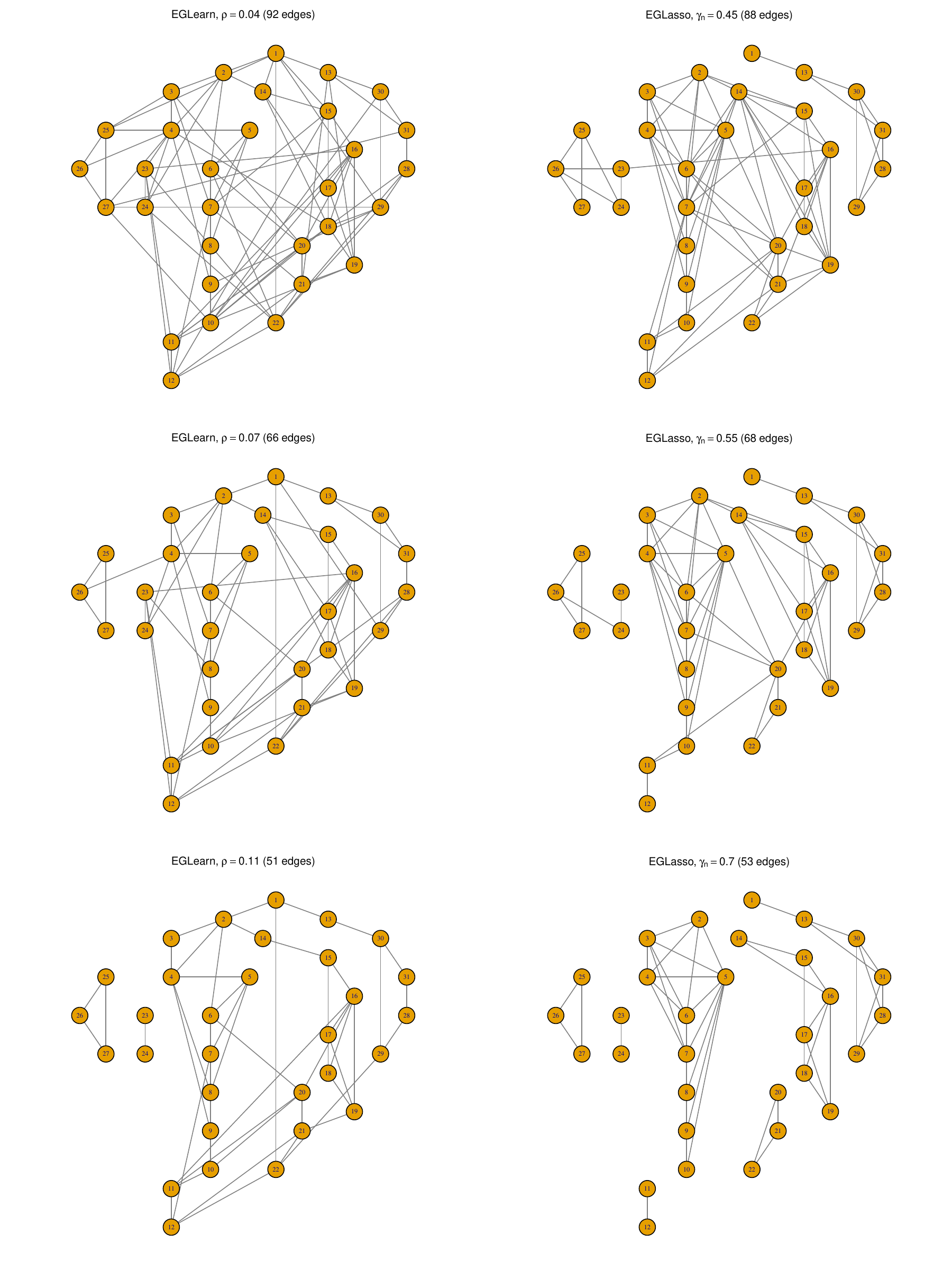}
	\caption{Danube river flow data: estimated graphical structures using EGLearn (left) and EGLasso (right) at three penalty levels. The penalties are chosen to yield a comparable number of edges.}
	\label{fig:danube application}
\end{figure}

\section*{Acknowledgements}

Phyllis Wan is supported by the Veni grant from the Dutch Research Council (VI.Veni.211E.034).  The authors would like to thank the associate editor and two anonymous referees for useful comments, and Frank R\"ottger and Rutger-Jan Lange for helpful discussions.


\bibliographystyle{chicago}
\bibliography{ref}


\appendix




\section{Proof of propositions in Sections~\ref{sec:hr} and \ref{sec:eglasso}} \label{appendix:prop:proofs}


\begin{proof}[Proof of Proposition~\ref{prop:variogram}]

Let $\Gamma$ be the variogram of a random vector $\mathbf{W}$ with positive definite covariance matrix $\Sigma$.  Then by definition $\Gamma$ is symmetric and has diagonal entries 0.   From \eqref{eq:sigma:gamma}, we can write
$$
	\Gamma = \text{diag}(\Sigma) \cdot \mathbf1^T + \mathbf1 \cdot \text{diag}(\Sigma) - 2\Sigma.
$$
For any $\mathbf{x}^T\mathbf1=0$,
$$
	\mathbf{x}^T \Gamma \mathbf{x} = \mathbf{x}^T\text{diag}(\Sigma) \cdot \mathbf1^T\mathbf{x} + \mathbf{x}^T\mathbf1 \cdot \text{diag}(\Sigma)\mathbf{x} - 2\mathbf{x}^T\Sigma\mathbf{x} = - 2\mathbf{x}^T\Sigma\mathbf{x}  < 0.
$$
Therefore $\Gamma  \in \mathcal{D}_0$.

Conversely, for any $\Gamma \in \mathcal{D}_0$, we will construct a positive definite covariance matrix such that the corresponding random vector has variogram $\Gamma$.  Let
$$
	\Sigma_k := \Gamma \mathbf{e}_k \mathbf1^T + \mathbf1\mathbf{e}_k^T \Gamma - \Gamma + \mathbf{e}_k \mathbf{e}_k^2.
$$
Any $\mathbf{y} \in \mathbb{R}^d$ can be written as
$$
	\mathbf{y} = \mathbf{x} + z\cdot \mathbf{e}_k
$$
for some $z \in \mathbb{R}$ and $\mathbf{x}^T\mathbf1=0$.  Then
\begin{eqnarray*}
	\mathbf{y}^T\Sigma_k\mathbf{y} &=& (\mathbf{x} + z\cdot  \mathbf{e}_k)^T(\Gamma \mathbf{e}_k \mathbf1^T + \mathbf1\mathbf{e}_k^T \Gamma - \Gamma + \mathbf{e}_k \mathbf{e}_k^T) (\mathbf{x} + z\cdot  \mathbf{e}_k) \\
	&=& (\mathbf{x} + z\cdot  \mathbf{e}_k)^T(\Gamma \mathbf{e}_k \mathbf1^T +  \mathbf1\mathbf{e}_k^T \Gamma - \Gamma) (\mathbf{x} + z\cdot  \mathbf{e}_k)  +  \left(\mathbf{e}_k^T (\mathbf{x} + z\cdot  \mathbf{e}_k)\right)^2 \\
	&=&  \mathbf{x} ^T(\Gamma \mathbf{e}_k \mathbf1^T + \mathbf1\mathbf{e}_k^T \Gamma - 2\Gamma ) \mathbf{x}  + z \cdot \mathbf{x}^T(\Gamma \mathbf{e}_k \mathbf1^T + \mathbf1\mathbf{e}_k^T \Gamma - 2\Gamma ) \mathbf{e}_k +  \\
	&& \quad + z\cdot \mathbf{e}_k^T(\Gamma \mathbf{e}_k \mathbf1^T + \mathbf1\mathbf{e}_k^T \Gamma - 2\Gamma ) \mathbf{x} + z^2 \cdot  \mathbf{e}_k^T(\Gamma \mathbf{e}_k \mathbf1^T + \mathbf1\mathbf{e}_k^T \Gamma - 2\Gamma ) \mathbf{e}_k+  \left(\mathbf{e}_k^T (\mathbf{x} + z\cdot\mathbf{e}_k)\right)^2 \\
	&=& -\mathbf{x}^T \Gamma \mathbf{x} +  \left(\mathbf{e}_k^T (\mathbf{x} + z\cdot\mathbf{e}_k)\right)^2 \\
	&\ge &0,
	\end{eqnarray*}
where the equality holds if and only if $\mathbf{x} = \mathbf0$ and $z=0$.  Therefore $\Sigma_k$ is a positive definite matrix.  It can be shown that
$$
	\Gamma = \text{diag}(\Sigma) \cdot \mathbf1^T + \mathbf1 \cdot \text{diag}(\Sigma) - 2\Sigma.
$$
Therefore $\Gamma$ is the variogram of a random vector with covariance matrix $\Sigma_k$.
\end{proof}

\begin{proof}[Proof of Proposition~\ref{prop:sigma:gamma}]

Since $\sum_{k=1}^d W_k' = \sum_{k=1}^d (W_k - \bar{W}) \equiv 0$, we have
$$
	\Sigma \cdot \mathbf1 = E\left(\sum_{k=1}^dW'_k \cdot \mathbf{W}'\right) = \mathbf0.
$$
From \eqref{eq:sigma:gamma}, we can write
$$
	\Gamma = \text{diag}(\Sigma) \cdot \mathbf1^T + \mathbf1 \cdot \text{diag}(\Sigma) - 2\Sigma.
$$
Therefore
\begin{eqnarray*}
	&& -\frac{1}{2}\left(I - \frac{\mathbf{1}\mathbf{1}^T}{d}\right) \Gamma \left(I - \frac{\mathbf{1}\mathbf{1}^T}{d}\right) \\
	&=& -\frac{1}{2}\left(I - \frac{\mathbf{1}\mathbf{1}^T}{d}\right) \left( \text{diag}(\Sigma) \cdot \mathbf1^T + \mathbf1 \cdot \text{diag}(\Sigma) - 2\Sigma \right)\left(I - \frac{\mathbf{1}\mathbf{1}^T}{d}\right) \\
	&=& -\frac{1}{2}\left(I - \frac{\mathbf{1}\mathbf{1}^T}{d}\right) \text{diag}(\Sigma) \cdot \mathbf1^T  \left(I - \frac{\mathbf{1}\mathbf{1}^T}{d}\right)  -\frac{1}{2}\left(I - \frac{\mathbf{1}\mathbf{1}^T}{d}\right) \mathbf1 \cdot \text{diag}(\Sigma) \left(I - \frac{\mathbf{1}\mathbf{1}^T}{d}\right) \\
	&& \quad +\left(I - \frac{\mathbf{1}\mathbf{1}^T}{d}\right) \Sigma \left(I - \frac{\mathbf{1}\mathbf{1}^T}{d}\right) \\
	&=&\left(I - \frac{\mathbf{1}\mathbf{1}^T}{d}\right) \Sigma \left(I - \frac{\mathbf{1}\mathbf{1}^T}{d}\right) \\
	&=& \Sigma,
\end{eqnarray*}
using the fact that $\mathbf1^T  \left(I - \frac{\mathbf{1}\mathbf{1}^T}{d}\right) = 0$ and $\Sigma \cdot \mathbf1=\mathbf0$.
\end{proof}

\begin{proof}[Proof of Proposition~\ref{prop:sigma:sigma_k}]

Let $\mathbf{W} \sim N(\mathbf0,\Sigma)$.  Then $\Sigma^{(k)}$ is the covariance matrix of 
$$
	\mathbf{W} - W_k \cdot \mathbf1 = (I - \mathbf1 \mathbf{e}_k) \cdot \mathbf{W}.
$$
From the linear transformation of Gaussian distributions,
$$
	\tilde\Sigma^{(k)} = (I - \mathbf1 \mathbf{e}_k^T) \cdot \Sigma \cdot (I - \mathbf1 \mathbf{e}_k^T)^T.
$$

For the reverse, we have
\begin{eqnarray*}
\frac{1}{d} \sum_{k=1}^d \tilde\Sigma^{(k)} &=& \Sigma - \frac{1}{d} \sum_{k=1}^d \mathbf1 \mathbf{e}_k^T\Sigma - \frac{1}{d} \sum_{k=1}^d \Sigma \mathbf{e}_k \mathbf1^T+  \frac{1}{d} \sum_{k=1}^d \mathbf1 \mathbf{e}_k^T\Sigma\mathbf{e}_k \mathbf1^T \\
&=& \Sigma -\mathbf1 \left( \frac{1}{d} \sum_{k=1}^d \mathbf{e}_k\right)^T\Sigma - \Sigma \left( \frac{1}{d} \sum_{k=1}^d \mathbf{e}_k\right) \mathbf1^T+   \mathbf1 \left(\frac{1}{d} \sum_{k=1}^d\mathbf{e}_k^T\Sigma\mathbf{e}_k\right) \mathbf1^T \\
&=& \Sigma -\frac{\mathbf1\mathbf1^T}{d}\Sigma - \Sigma \frac{\mathbf1\mathbf1^T}{d} +   \mathbf1 \left(\frac{1}{d} tr(\Sigma)\right) \mathbf1^T \\
&=& \Sigma + \frac{1}{d} tr(\Sigma) \cdot \mathbf1\mathbf1^T.
\end{eqnarray*} 
Therefore
$$
	\Sigma = \frac{1}{d} \sum_{k=1}^d \tilde\Sigma^{(k)} - \frac{1}{d} tr(\Sigma) \cdot \mathbf1\mathbf1^T.
$$

\end{proof}


\begin{proof}[Proof of Proposition~\ref{prop:sigma:theta}]

From $\Sigma\cdot \mathbf1 = \mathbf0$, we have that $\frac{1}{\sqrt{d}}\cdot\mathbf1$ is an eigenvector of $\Sigma$ with eigenvalue $0$.  Since $\Sigma$ is positive semidefinite with rank $d-1$. it follows that the rest of the eigenvalues $\lambda_2\le \cdots \le \lambda_d$ must be positive.

Next, we will show that
$$
	\left(\Sigma + \frac{M}{d}\mathbf1\mathbf1^T\right) \cdot \left( \Theta + \frac{1}{dM} \mathbf1\mathbf1^T\right) = I.
$$
Note that
$$
	\tilde\Sigma^{(k)} \cdot \Theta 
	=
	 I - \mathbf{e}_k \mathbf1^T.
$$
(Consider the example of $k=d$, where
$$
	\tilde\Sigma^{(d)} \cdot \Theta =
	\begin{pmatrix}
	\Sigma^{(d)} & \mathbf0 \\
	\mathbf0^T & 0
	\end{pmatrix}
	\cdot
	\begin{pmatrix}
	\Theta^{(d)} & \Theta^{(d)}\mathbf1 \\
	\mathbf1^T \Theta^{(d)} & \mathbf1^T \Theta^{(d)}\mathbf1
	\end{pmatrix}
	=
	 I - \mathbf{e}_d \mathbf1^T.
$$
Other cases of $k$ can be generalized accordingly.)  Then from Proposition~\ref{prop:sigma:sigma_k},
\begin{eqnarray*}
	\Sigma \cdot \Theta &=& \left(\frac{1}{d} \sum_{k=1}^d \tilde\Sigma^{(k)} + \frac{1}{d} tr(\Sigma) \cdot  \mathbf1\mathbf1^T \right) \cdot \Theta \\
	&=& \frac{1}{d} \sum_{k=1}^d \tilde\Sigma^{(k)} \cdot \Theta + \frac{1}{d} tr(\Sigma) \cdot  \mathbf1\mathbf1^T \cdot \Theta \\
	&=& \frac{1}{d} \sum_{k=1}^d \left(I - \mathbf{e}_k \mathbf1^T\right) \\
	&=& I - \frac{1}{d} \mathbf1\mathbf1^T.
\end{eqnarray*}
Therefore
$$
	\left(\Sigma + \frac{M}{d}\mathbf1\mathbf1^T\right) \cdot \left( \Theta + \frac{1}{dM} \mathbf1\mathbf1^T\right) = \Sigma \cdot \Theta + \frac{M}{d}\mathbf1\mathbf1^T \cdot \frac{1}{dM} \mathbf1\mathbf1^T = I - \frac{1}{d} \mathbf1\mathbf1^T +  \frac{1}{d} \mathbf1\mathbf1^T = I.
$$

Now $\Sigma + \frac{M}{d}\mathbf1\mathbf1^T$ is positive definite and has the eigendecomposition:
$$
	\Sigma + \frac{M}{d}\mathbf1\mathbf1^T = \sum_{k=2}^d \lambda_k \mathbf{u}_k\mathbf{u}_k^T + M \cdot \mathbf{u}_1\mathbf{u}_1^T.
$$
Therefore its inverse as the eigendecomposition:
$$
	\Theta + \frac{1}{dM} \mathbf1\mathbf1^T = \sum_{k=2}^d \frac{1}{\lambda_k} \mathbf{u}_k\mathbf{u}_k^T + \frac{1}{M} \cdot \mathbf{u}_1\mathbf{u}_1^T.
$$
It follows that
$$
	\Theta = \sum_{k=2}^d \frac{1}{\lambda_k} \mathbf{u}_k\mathbf{u}_k^T.
$$
And the proposition is proved.

\end{proof}


\begin{proof}[Proof of Proposition~\ref{prop:mle}]

We proceed to show two results,
\begin{equation} \label{eq:mle:eq:1}
	\frac{1}{d}\sum_{k=1}^d \log|\Theta^{(k)}| = \log|\Theta|_+ + \log(d),
\end{equation}
and
\begin{equation} \label{eq:mle:eq:2}
	\frac{1}{d}\sum_{k=1}^d tr\left(\hat{\Sigma}^{(k)} \Theta^{(k)}\right) = tr\left(S \Theta\right).
\end{equation}

First consider \eqref{eq:mle:eq:1}.  
Consider $d=k$.  We have the following transformation
\begin{eqnarray*}
	&&
	\left(
		\begin{array}{cc}
			I& {\bf0} \\
			-\mathbf{1}^T & 1
		\end{array}
	\right)
	\left(
		\begin{array}{cc}
			I& \frac{1}{d}\mathbf{1} \\
			{\bf0}^T & 1
		\end{array}
	\right)
	\left(
		\begin{array}{cc}
			\Theta^{(d)} & {\bf0} \\
			{\bf0}^T & d
		\end{array}
	\right)
	\left(
		\begin{array}{cc}
			I& {\bf0} \\
			\frac{1}{d}\mathbf{1}^T & 1
		\end{array}
	\right)
	\left(
		\begin{array}{cc}
			I& -\mathbf{1} \\
			{\bf0}^T & 1
		\end{array}
	\right) \\
	&=&
	\left(
		\begin{array}{cc}
			I& {\bf0} \\
			-\mathbf{1}^T & 1
		\end{array}
	\right)
	\left(
		\begin{array}{cc}
			\Theta^{(d)} + \frac{1}{d}\mathbf{1}\mathbf{1}^T & {\bf1} \\
			{\bf1}^T & d
		\end{array}
	\right)
	\left(
		\begin{array}{cc}
			I& -\mathbf{1} \\
			{\bf0}^T & 1
		\end{array}
	\right) \\
	&=&
	\left(
		\begin{array}{cc}
			\Theta^{(d)} + \frac{1}{d}\mathbf{1}\mathbf{1}^T & -\Theta^{(d)}\mathbf{1} + \frac{1}{d}{\bf1} \\
			-\mathbf{1}^T\Theta^{(d)} + \frac{1}{d}{\bf1}^T & \mathbf{1}^T\Theta^{(d)}\mathbf{1} + \frac{1}{d}
		\end{array}
	\right) \\
	&=& \Theta + \frac{1}{d}\mathbf{1}\mathbf{1}^T 
\end{eqnarray*}
Since
$$
	\left|
	\left(
		\begin{array}{cc}
			I& {\bf0} \\
			-\mathbf{1}^T & 1
		\end{array}
	\right)
	\right|
	=
	\left|
	\left(
		\begin{array}{cc}
			I& {\bf0} \\
			\frac{1}{d}\mathbf{1}^T & 1
		\end{array}
	\right)
	\right|
	=1,
$$
we have
$$
	|\Theta|_+ = \left| \Theta + \frac{1}{d}\mathbf{1}\mathbf{1}^T\right| = 
	\left|\left(
	\begin{array}{cc}
	\Theta^{(d)} & {\bf0} \\
	{\bf0}^T & d
	\end{array}
	\right)\right|
	= d |\Theta^{(d)}|.
$$
Similarly, we can show that for any $k$,
$$
	|\Theta|_+ = d |\Theta^{(k)}|.
$$
This leads to \eqref{eq:mle:eq:1}.

Now consider \eqref{eq:mle:eq:2}.  We have
\begin{eqnarray*}
\frac{1}{d}\sum_{k=1}^d tr\left({\hat\Sigma}^{(k)}\Theta^{(k)}\right) &=&  tr\left(\left( \frac{1}{d}\sum_{k=1}^d\tilde{\hat\Sigma}^{(k)}\right) \Theta\right) \\
&=&  tr\left(\left( \frac{1}{d}\sum_{k=1}^d\tilde{\hat\Sigma}^{(k)} - \left(\frac{1}{d^3}\sum_{k=1}^d {\mathbf1^T\tilde{\hat\Sigma}^{(k)}\mathbf1}\right)\mathbf1\mathbf1^T\right) \Theta \right) \\
&=& tr(S\Theta).
\end{eqnarray*}

And the proposition is proved.

\end{proof}


\section{Proof of Theorem~\ref{thm:main}} \label{appendix: proof for Section 4}


The proof of Theorem~\ref{thm:main} follows the same strategy as that of Theorem~1 of \cite{ravikumar2011high}. The key difference arises from the shifted penalty target $c>0$ in the extreme graphical lasso, which affects the construction of the witness solution and requires a modified analysis of the continuity of the mapping $F$. We highlight these differences as they arise in the proof.

\begin{proof}[Proof of Theorem~\ref{thm:main}]
The proof is constructed by adapting  the proof of Theorem 1 in \cite{ravikumar2011high} to our setting: shrinking the off-diagonal elements to a given constant $c$. We point out the differences whenever needed.

We shall focus on the event
$$A=\{\|S-\Sigma\|_\infty< \delta_n\},$$.

Recall that $c = \frac{1}{dM}$,
$$
	\Theta^* = \Theta + c \mathbf1\mathbf1^T,
$$
and  $\hat\Theta^*$ is the solution to the following graphical lasso problem
$$
	\hat{\Theta}^* := {\arg\min}_{\Theta^*} \left\{-\log|\Theta^*| + tr(S^*\Theta^*) + \gamma_n \sum_{i\ne j} |\Theta^*_{ij} - c|\right\}
$$
and $\hat{\Theta}_{lasso} := \hat{\Theta}^* - c\mathbf1\mathbf1^T$.   The estimated edge set is $\hat E:=\{(i,j):\hat\Theta_{lasso,ij} \neq 0\}.$ 
We start by proving that on the even $A$ such a solution exists, is unique and satisfies that, $\hat E \subset E,$ and
$$
	\|\hat{\Theta}_{lasso} - \Theta\|_\infty \le r_n,
$$
which is equivalent to proving
$$
	\|\hat{\Theta}^* - \Theta^*\|_\infty \le r_n.
$$


We first show that the solution $\hat\Theta^*$ exsits and is unique. The proof follows the same lines as that for Lemma 3 in \cite{ravikumar2011high}. The estimator $S^*$ is positive definite with all diagonal elements being positive. The rest of the proof follows exactly the same arguments therein.

Next, the solution $\hat\Theta^*$ must satisfy the following KKT condition.
\begin{equation} \label{eq:KKT for Theta}
	-\left(\hat\Theta^*\right)^{-1} + S^* + \gamma_n\hat{Z} = 0,
\end{equation}
where
$$
	\hat Z_{ij} = \begin{cases}
		0 & \text{  if } i=j,\\
		\text{sign}(\hat\Theta^*_{ij}- c) & \text{  if } i\neq j \text{ and } \hat\Theta^*_{ij} \ne c,\\
		\in[-1,1] &\text{  if } i\neq j \text{ and } \hat\Theta^*_{ij}= c.
	\end{cases}
$$

In the following proof, we follow the same idea as in \cite{ravikumar2011high}, but adapting that to this modified KKT condition.
We shall construct a ``witness'' precision matrix $\tilde{\Theta}^*$ as follows. Let $\tilde\Theta^*$ be the solution to the following optimization problem,
\begin{equation} \label{eq: problem tilde Theta}
	\tilde{\Theta}^* := {\arg\min}_{\{\Theta^*:\Theta^*_{ij} = c, (i,j) \in E^c\}} -\log|\Theta^*| + tr(S^*\Theta^*) + \gamma_n \sum_{i\ne j} |\Theta^*_{ij} - c|.
\end{equation}
This is the same optimization but constrained to a smaller domain.  Let $\tilde{E}$ denote the graph recovered from $\tilde\Theta^*$. Clearly, $\tilde{\Theta}^*$ satisfies: $\tilde{\Theta}^*_{ij} = c$ for $(i,j) \in E^c$, i.e. $\tilde{E} \subset E$.

We shall show that under the conditions in Theorem \ref{thm:main},
\begin{itemize}
\item
	$\tilde{\Theta}^*$ satisfies the above KKT condition;
\item
	$\|\tilde{\Theta}^* - \Theta^*\|_\infty \le r_n$.
\end{itemize}
Then by uniqueness, $\tilde{\Theta}^* = \hat\Theta^*$ which achieves the goal that we are aiming to prove.


With a similar argument regarding the existence and uniqueness of the original optimization problem, the solution to the problem \eqref{eq: problem tilde Theta} also exists and is unique. In addition, it satisfies a similar KKT condition as follows,
$$
	-\left(\tilde\Theta^*\right)^{-1}_{ij} + S^*_{ij} + \gamma_n\tilde{Z}_{ij} = 0, \quad (i,j) \in E,
$$
where
$$
	\tilde Z_{ij} = \begin{cases}
		\text{sign}(\tilde\Theta^*_{ij}- c) & \text{if } \tilde\Theta^*_{ij} \ne c,\\
		\in[-1,1] &\text{if } \tilde\Theta^*_{ij}= c.
	\end{cases}
$$
This coincides with the KKT condition \eqref{eq:KKT for Theta}, but the restrictions are only on entries indexed by $E$.  As a matter of fact, $\tilde{Z}$ is only defined on $E$.  In order to argue $\tilde{\Theta}^*$ as a candidate for $\hat\Theta^*$ and satisfies the full KKT condition, we will now extend the definition of $\tilde{Z}$ to $E^c$ as well.

Define
$$
	\tilde{Z}_{ij} := \frac{1}{\gamma_n} \left(\left(\tilde\Theta^*\right)^{-1}_{ij} - S^*_{ij}\right), \quad (i,j) \notin E.
$$
Then the pair $(\tilde\Theta^*,\tilde{Z})$ satisfies the original KKT equation \eqref{eq:KKT for Theta}.  What remains to be proved is that $\tilde{Z}$ also satisfies
$$
	|\tilde{Z}_{ij}| \le 1, \quad (i,j) \notin E.
$$

To summarize, in order to complete the proof of Theorem \ref{thm:main}, we will show that on the event $A$,\\
{\bf Goal 1}: $|\tilde{Z}_{ij}| \le 1, \quad (i,j) \notin E.$\\
{\bf Goal 2}: $\|\tilde{\Theta}^* - \Theta^*\|_\infty \le r_n$.

In the rest of the proof we denote
$$
	\Delta := \tilde\Theta^* - \Theta^*.
$$
For $(i,j) \notin E$, $\tilde\Theta_{ij} = c$ by definition and $\Theta_{ij} = c$.  Therefore, $\Delta_{E^c} = \mathbf0$ and {\bf Goal 2} above can be translated to
$$
	\|\Delta_E\| _\infty \le r_n.
$$	

To handle the KKT condition for $\tilde\Theta^*$, we start with handling $\left(\tilde\Theta^*\right)^{-1}$ as follows:
\begin{eqnarray*}
	\left(\tilde\Theta^*\right)^{-1} &=& \left(\Theta^* + \Delta\right)^{-1} \\
	&=&  \left(\Theta^*(I  + \Sigma^*\Delta)\right)^{-1} \\
	&=& (I  + \Sigma^*\Delta)^{-1} \Sigma^* =: J\Sigma^*.
\end{eqnarray*}
Now in the case where $|||\Sigma^*\Delta|||_\infty < 1$, we can expand $J$ as
$$
	J = \sum_{k=0}^\infty (-1)^k(\Sigma^*\Delta)^k = I - \Sigma^*\Delta + (\Sigma^*\Delta)^2J.
$$
Inspired from this relation, we can use $\Sigma^* - \Sigma^*\Delta\Sigma^*$ to approximate $\left(\tilde\Theta^*\right)^{-1}$ and define
$$
	R :=\left(\tilde\Theta^*\right)^{-1} - (\Sigma^* - \Sigma^*\Delta\Sigma^*),
$$
as the approximation error. Here $R$ is defined regardless of whether $|||\Sigma^*\Delta|||_\infty < 1$.

Recall that $\Sigma^* = \Sigma + \frac{M}{d}\mathbf1\mathbf1^T$ and $S^* = S + \frac{M}{d}\mathbf1\mathbf1^T$. Define
$$
	 R':=S^*-\Sigma^*=S-\Sigma.
$$
On the event $A$, we have that $\|R'\|_\infty \le \delta_n$.  

Rewrite the KKT condition as
$$
	\Sigma^*\Delta\Sigma^* - R + R' + \gamma_n\tilde{Z} = 0.
$$
We vectorize it using the notation $\bar{\cdot}$ as the vectorization of a matrix.  Then the vectorized KKT condition is 
$$
	\overline{\Sigma^*\Delta\Sigma^*} - \overline R + \overline{R'} + \gamma_n\overline{\tilde{Z}} = 0.
$$
Note that
$$
	\overline{\Sigma^*\Delta\Sigma^*} = (\Sigma^*\otimes\Sigma^*)\overline\Delta =: \Omega\overline\Delta,
$$
where $\Omega:= \Sigma^*\otimes\Sigma^*$ denotes the Kronecker product of $\Sigma^*$ with itself.  Then we have
$$
	\Omega\overline\Delta - \overline R + \overline{R'} + \gamma_n\overline{\tilde{Z}} = 0.
$$

By examining the rows of $\Omega$ indexed by $E$ and $E^c$ separately and noting that $\Delta_{E^c}=0$, we get
\begin{eqnarray}
	\Omega_{EE}\overline\Delta_E - \overline R_E + \overline{R'}_E + \gamma_n\overline{\tilde{Z}}_E = 0, \label{eq: first equation with E rows}\\
	\Omega_{E^cE}\overline\Delta_E - \overline R_{E^c} + \overline{R'}_{E^c} + \gamma_n\overline{\tilde{Z}}_{E^c} = 0. \label{eq: second equation with Ec rows}
\end{eqnarray}


\subsection*{Proof of Goal 2}

To prove Goal 2, we shall show that for any pre-specified $\|\tilde{Z}_E\|_\infty \le 1$, there exists a solution to $\Delta$ which satisfies: 
\begin{itemize}
\item
	$-\left(\Theta^* + \Delta\right)^{-1}_E + S^*_E+ \gamma_n\tilde{Z}_E = 0$;
\item
	$\Delta_{E^c} = \mathbf0$;
\item
	$\|\Delta_E\|_\infty \le r_n$.
\end{itemize}

With the statement above proven, given the $\tilde{Z}_E$ produced from the KKT condition for $\tilde\Theta^*$, a solution $\Delta$ exists and coincides with $\tilde Z_E$. Then this is the unique solution.  Hence $\|\tilde\Theta^* - \Theta^*\|_\infty \le r_n$ which concludes {\bf Goal 2}.

Now we construct such a solution $\Delta$.  Recall that $\Delta_{E^c} = \mathbf0$, we only need to construct a suitable $\Delta_E$ by utilizing the Brouwer fixed point theorem.  

The solution $\Delta_E$ satisfies \eqref{eq: first equation with E rows}, which can be rewritten as
$$
	\overline\Delta_E  = (\Omega_{EE}) ^{-1} \left(\overline R_E - \overline{R'}_E - \gamma_n\overline{\tilde{Z}}_E\right).
$$
We regard
$$
	R =\left(\tilde\Theta^*\right)^{-1} - (\Sigma^* - \Sigma^*\Delta\Sigma^*)
$$
as a function of $\Delta$ or eventually a function of $\overline\Delta_E$.  To stress this point we define it as $R = R(\overline\Delta_E)$. Also recall that $R' = S^* - \Sigma^*$ does not depend on $\overline\Delta_E$.  Then we can write the above equation as 
$$
	\overline\Delta_E  = (\Omega_{EE})^{-1} \left(\overline R_E(\overline\Delta_E) - \overline{R'}_E - \gamma_n\overline{\tilde{Z}}_E\right): = F(\overline\Delta_{E}).
$$

Consider the closed ball $B(r_n) ;= \{x \in \mathbb{R}^{|E|}: \|x\|_\infty \le r_n\}$.   If $F$ is a continuous mapping from $B(r_n)$ onto itself, then there exists a fixed point $\overline\Delta_E$ on $B(r_n)$ such that $\overline\Delta_E = F(\overline\Delta_{E})$ following the Brouwer fixed point theorem. This is exactly the desired solution. 

In the proof of Theorem 1 in \cite{ravikumar2011high}, a similar mapping $F$ was defined and claimed to be ``clearly continuous''. In fact, the continuity requires more conditions. Here we provide a careful argument for the continuity.

While most terms in the definition of $F$ are continuous functions, the term $\overline{\tilde{Z}}_E$ is related to $\overline{\Delta}_E$ by $\tilde{Z}_E=\text{sign}(\tilde{\Theta}^*_E-c)=\text{sign}(\Delta_E+\Theta^*_E-c)$,  which is a potentially discontinuous function, due to the discontinuity of the sign function. Nevertheless, with assuming that $\min\{\|\Theta_{ij}\|; (i,j) \in E, i\neq j\}>r_n$, since $\|\Delta_E\|\leq r_n$, we have $\tilde{Z}_E=\text{sign}(\Delta_E+\Theta^*_E-c)=\text{sign}(\Theta_E)$, which is not related to $\Delta_E$. Then $F$ defined above is a continuous mapping.

Next, we need to show that $F$ projects $B(r_n)$ onto itself, that is, for any $\overline\Delta_E$ satisfying $\|\overline\Delta_E\|_\infty \le r_n$, we have $\|F(\overline\Delta_E)\|_\infty \le r_n$. For any $\|\overline\Delta_E\|_\infty \le r_n$, we write
$$
	\|F(\overline\Delta_E)\|_\infty \le  |||(\Omega_{EE})^{-1}|||_\infty \left(  \|R\|_\infty + \|R'\|_\infty + \gamma_n \|\overline{\tilde{Z}}_E\|_\infty \right)\le  |||(\Omega_{EE}) ^{-1}|||_\infty \left(  \|R\|_\infty + \|R'\|_\infty + \gamma_n \right),
$$
due to the fact that $\|\overline{\tilde{Z}}_E\|_\infty \le 1$

We first handle $\|R\|_\infty$. Recall that 
$$
	R :=\left(\tilde\Theta^*\right)^{-1} - (\Sigma^* - \Sigma^*\Delta\Sigma^*)= \sum_{k=2}^\infty (-1)^k(\Sigma^*\Delta)^k \Sigma^*  = (\Sigma^*\Delta)^2J\Sigma^*,
$$
where 
$$
	J = (I + \Sigma^* \Delta)^{-1} = \sum_{k=0}^\infty (-1)^k(\Sigma^*\Delta)^k,
$$
provided that
$||| \Sigma^* \Delta|||_\infty < 1 $. To ensure this condition, 
$$
|||\Sigma^*\Delta|||_\infty = |||\Sigma^* (\tilde\Theta - \Theta)|||_\infty \le D|||\Sigma^*|||_\infty\cdot r_n,
$$
where $D$ is the maximum degree in the graph. Therefore $||| \Sigma^* \Delta|||_\infty < 1$ holds by  requiring that
\begin{equation} \label{eq: required condition 1}
	D|||\Sigma^*|||_\infty\cdot r_n \le C_4< 1.
\end{equation}
The upper bound $C_4$ implies that $
	|||J^T|||_\infty \leq \frac{1}{1-C_4}.$

With the condition \eqref{eq: required condition 1}, we can further derive an upper bound for $\|R\|_\infty$. Consider one specific element in $R$. With denoting  $\mathbf{e}_i$ as a vector with all zero elements except a one element at the $i-$th dimension, we have that
$$R_{ij}=\mathbf{e}_i^T (\Sigma^*\Delta)^2J\Sigma^* \mathbf{e}_j\leq \|\mathbf{e}_i^T(\Sigma^*\Delta)^2\|_{\infty}\| J\Sigma^* \mathbf{e}_j\|_{1}\leq \|(\Sigma^*\Delta)^2\|_\infty|||\Sigma^* J^T|||_\infty .$$
By considering all possible $(i,j)$ we get that,
\begin{eqnarray*}
 \|R\|_\infty &\le& \|(\Sigma^*\Delta)^2\|_\infty|||\Sigma^* J^T|||_\infty\\
 &\le& |||\Sigma^*\Delta|||_\infty \|\Sigma^*\Delta\|_\infty|||J^T|||_\infty  |||\Sigma^*|||_\infty\\
&\le& D|||\Sigma^*|||_\infty\cdot r_n \cdot |||\Sigma^*|||_\infty r_n \cdot \frac{1}{1-C_4}\cdot |||\Sigma^*|||_\infty\\
&<& C_5 \cdot r_n^2,
\end{eqnarray*}
where $C_5=\frac{D}{1-C_4}|||\Sigma^*|||_\infty^3 $.

Next, since $
	R' = S^* - \Sigma^* = S - \Sigma,$ we have that on the event $A$, $
	\|R'\|_\infty \le \delta_n.
$ Combining the upper bounds for $\|R\|_\infty$ and $\|R'\|_\infty$, we get that,
$$
	\|F(\Delta_E)\|_\infty  \le  |||(\Omega_{EE})^{-1}|||_\infty ( C_5 \cdot r^2_n + \delta_n + \gamma_n) \le r_n,
$$
by requiring that
\begin{equation} \label{eq: required condition 2}
	C_5 \cdot r_n^2 + \delta_n + \gamma_n \le \frac{1}{|||(\Omega_{EE})^{-1}|||_\infty} \cdot r_n
\end{equation}

If the two required conditions \eqref{eq: required condition 1} and \eqref{eq: required condition 2} hold, we achieve \textbf{Goal 2} by utilizing the Brouwer fixed point theorem.

\subsection*{Proof of Goal 1}

To prove \textbf{Goal 1}, we shall show that with the constructed solution above, we have
$$
	\|\overline{\tilde{Z}}_{E^c}\|_\infty \le 1.
$$

We rewrite the equation \eqref{eq: second equation with Ec rows} as
$$
	\overline{\tilde{Z}}_{E^c} = -\frac{1}{\gamma_n}\Omega_{EE^c}\overline\Delta_E + \frac{1}{\gamma_n} \overline R_{E^c} -\frac{1}{\gamma_n} \overline{R'}_{E^c},
$$
and the substitute  $\overline\Delta_E$ above using \eqref{eq: first equation with E rows} to get that 
$$
	\overline{\tilde{Z}}_{E^c} = -\frac{1}{\gamma_n}\Omega_{EE^c}(\Omega_{EE})^{-1} \left(- \overline R_E + \overline{R'}_E + \gamma_n\overline{\tilde{Z}}_E\right) + \frac{1}{\gamma_n} \overline R_{E^c} -\frac{1}{\gamma_n} \overline{R'}_{E^c}.
$$
The upper bound for $\|\overline{\tilde{Z}}_{E^c}\|_\infty $ is then  
$$
	\|\overline{\tilde{Z}}_{E^c}\|_\infty \le \frac{1}{\gamma_n}|||\Omega_{EE^c}(\Omega_{EE})^{-1}|||_\infty \left( \|R\|_\infty + \|R'\|_\infty +\gamma_n\|\overline{\tilde{Z}}_E\|_\infty \right) + \frac{1}{\gamma_n}\|R\|_\infty + \frac{1}{\gamma_n}\|R'\|_\infty.
$$
Using the same upper bounds derived in the proof of \textbf{Goal 2}, we have
\begin{eqnarray*}
	\|\overline{\tilde{Z}}_{E^c}\|_\infty &\le& \frac{1}{\gamma_n}|||\Omega_{EE^c}(\Omega_{EE})^{-1}|||_\infty \left( \delta_n + C_5 \cdot r_n^2 +\gamma_n\right) +  \frac{1}{\gamma_n}(\delta_n + C_5 \cdot r_n^2)\\
	&=&|||\Omega_{EE^c}(\Omega_{EE})^{-1}|||_\infty+\frac{1}{\gamma_n}\left(|||\Omega_{EE^c}(\Omega_{EE})^{-1}|||_\infty+1\right)\left( \delta_n + C_5 \cdot r_n^2\right).
\end{eqnarray*}
Since Condition \eqref{con:mi} ensures that $|||\Omega_{EE^c}(\Omega_{EE})^{-1}|||_\infty<1-\alpha$, 
To satisfy $\|\overline{\tilde{Z}}_{E^c}\|_\infty \le 1$, we only need to further require
\begin{equation} \label{eq: required condition 3}
	\delta_n + C_5 \cdot r_n^2 \le \frac{\alpha}{1-\alpha} \gamma_n.	
\end{equation}

To conclude, the theorem is proven provided that the three conditions \eqref{eq: required condition 1}--\eqref{eq: required condition 3} hold. The last step is to verify these three relations under the conditions in Theorem \ref{thm:main}.

Recall that $r_n$ is defined in \eqref{eq:definition r}
$$r_n:=\frac{|||\left(\Omega_{EE} \right)^{-1}|||_\infty}{1-\alpha} \cdot \gamma_n.$$
Clearly, this definition together with  \eqref{eq: required condition 3} implies  \eqref{eq: required condition 2}. Hence we only need to verify the conditions \eqref{eq: required condition 1} and  \eqref{eq: required condition 3}.

We write the two conditions in terms of  $\delta_n$ and $\gamma_n$:
\begin{eqnarray*}
	D|||\Sigma^*|||_\infty \frac{|||(\Omega_{EE})^{-1}|||_\infty}{1-\alpha} \cdot \gamma_n  &\le&C_4 \\
	\delta_n + \frac{D}{1-C_4}|||\Sigma^*|||_\infty^3 \cdot \left(\frac{|||(\Omega_{EE})^{-1}|||_\infty}{1-\alpha} \right)^2\cdot \gamma_n^2 & \le& \frac{\alpha}{1-\alpha} \cdot \gamma_n.
\end{eqnarray*}
where $C_5$ is substituted by $ \frac{D}{1-C_4}|||\Sigma^*|||_\infty^3$:

The lower bound for $\gamma_n$ in \eqref{eq:lowerbound rho} ensures that $\delta_n\leq \epsilon \frac{\alpha}{1-\alpha} \gamma_n$ for some $0<\epsilon<1$, we thus need to require that
$$
\frac{D}{1-C_4}|||\Sigma^*|||_\infty^3 \cdot \left(\frac{|||(\Omega_{EE})^{-1}|||_\infty}{1-\alpha} \right)^2\cdot \gamma_n\leq (1-\epsilon) \frac{\alpha}{1-\alpha},
$$
which guarantees the second condition. Together with the first condition, we have obtained an upper bound for $\gamma_n$ as
$$\gamma_n\leq \min\left\{\frac{C_4(1-\alpha)}{D|||\Sigma^*|||_\infty |||(\Omega_{EE})^{-1}|||_\infty},\frac{(1-C_4)(1-\epsilon)\alpha(1-\alpha)}{D|||\Sigma^*|||_\infty^3|||(\Omega_{EE})^{-1}|||_\infty^2}\right\}.$$
We choose $C_4$ such that the two terms in the minimum are equal. That is
$$C_4=\frac{(1-\epsilon)\alpha}{(1-\epsilon)\alpha+|||\Sigma^*|||_\infty^2|||(\Omega_{EE})^{-1}|||_\infty}<1,$$
which leads to
$$\gamma_n\leq \frac{(1-\epsilon)\alpha(1-\alpha)}{D|||\Sigma^*|||_\infty |||(\Omega_{EE})^{-1}|||_\infty\left[(1-\epsilon)\alpha+|||\Sigma^*|||_\infty^2|||(\Omega_{EE})^{-1}|||_\infty\right]}.$$
This is exactly the required upper bound for $\gamma_n$ in \eqref{eq:upperbound rho}.


\end{proof}


\section{Proof of Proposition~\ref{prop:accuracy of S estimation}} \label{appendix:proof of S estimation}

\begin{proof}[Proof of Proposition \ref{prop:accuracy of S estimation}] We intend to apply Theorem 3 in \cite{engelke2021learning}. For that purpose, we first verify all assumptions needed for that theorem, namely Assumptions 1 and 2 therein.

We handle Assumption 1 first. Based on Lemma S3 in \cite{engelke2021learning}, Condition \ref{assumption:non-degenerate} implies that for any $\xi''>0$, there exists $K_{\xi''}>0$ depending on $\underline{\lambda}$ and $\overline{\lambda}$, but independent of $d$, such that Assumption 4 therein holds. 
Denote $K=K'+2K_{\xi''}$ and $\xi=\xi' \xi''/(1+\xi'+\xi'')$. Together with Condition \ref{assumption:second order}, we get that Assumption 1 in \cite{engelke2021learning} holds. In particular, one can choose $\xi''$ sufficiently large such that $\xi$ can be any constant satisfying $\xi<\xi'$. 
	
Next, Assumption 2 in \cite{engelke2021learning} holds automatically for all non-degenerate HR distribution satisfying our Condition \ref{assumption:non-degenerate}. Therefore, we can then apply Theorem 3 therein to obtain that there exists positive constants $C_1$, $C_2$ and $C_3$, independent of $d$, such that for any $k_n\geq n^{\xi}$ and $\lambda\leq \sqrt{k_n}/(\log n)^4$,
	\begin{equation}\label{eq:inequality for sigma_k}
	P\left(\max_{1\leq k\leq d}\|\hat\Sigma^{(k)} -\Sigma^{(k)}\|_\infty> C_1\left\{\left(\frac{k_n}{n}\right)^{\xi}\left(\log\left(\frac{k_n}{n}\right)\right)^2+\frac{1+\lambda}{\sqrt{k_n}}\right\}\right)\leq C_2 d^3e^{-C_3\lambda^2}.
	\end{equation}
	The constant $C_1$ here equals to $\frac{3}{2}\overline{C}$ in Theorem 3 in \cite{engelke2021learning} because we are estimating the matrix $\Sigma$ instead of the variogram $\Gamma$.
	
	For any $\varepsilon\geq C_2 d^3\exp\{-\frac{C_3k_n}{(\log n)^8}\}$, one can choose $\lambda=\sqrt{\frac{1}{C_3}\log (C_2 d^3/\varepsilon)}\leq \frac{\sqrt{k_n}}{(\log n)^4}$ in \eqref{eq:inequality for sigma_k} to obtain the element-wise bound for the estimation error $\hat\Sigma^{(k)} -\Sigma^{(k)}$ uniformly for all $1\leq k\leq d$.
	
	Since 
	$$S- \Sigma = \frac{1}{d}\sum_{k=1}^d(\hat\Sigma^{(k)}-\Sigma^{(k)})-\frac{1}{d}\sum_{k=1}^d\left(\frac{1}{d^2}\mathbf{1}^T(\hat\Sigma^{(k)}-\Sigma^{(k)})\mathbf{1}\right)\mathbf{1}\mathbf{1}^T,$$
	which implies that $\|S-\Sigma\|_\infty\leq 2\max_{1\leq k\leq d}\|\hat\Sigma^{(k)} -\Sigma^{(k)}\|_\infty$,
	we immediately get the inequality \eqref{eq:inequality for sigma} with replacing $C_1$ by $2C_1$. W.l.o.g., we continue using $C_1$.
	
	For the asymptotic statement, if $(\log n)^4\sqrt{\frac{\log d}{k_n}}\to 0$ as $n\to\infty$, then the lower bound for $\varepsilon$, $\varepsilon_n\to 0$ as $n\to\infty$. In other words, for any $\varepsilon>0$, with sufficiently large $n$, $\varepsilon>\varepsilon_n$ holds. The asymptotic statement follows immediately.
	\end{proof}

\section{Proof of Proposition~\ref{prop:mi:upper:bound}} \label{app:proof:mi:upperbound}

To prove Proposition~\ref{prop:mi:upper:bound}, we require the following lemma. 

\begin{lemma} \label{lemma:2}
Let $A \in \mathbb{R}^{p\times p}$ be a symmetric square martix and let $A_{EE}$ be a principal sub-matrix of $A$ obtained by only keeping the columns and rows indexed by $E \subset \{1,\ldots,p\}$.  Then the largest and smallest eigenvalues of $A_{EE}$ can be bounded by
$$
	\lambda_{\min}(A_{EE}) \ge \lambda_{\min} (A), \quad \lambda_{\max}(A_{EE}) \le \lambda_{\max}(A).
$$

\end{lemma}

\begin{proof}[Proof of Lemma~\ref{lemma:2}]
By definition,
$$
    \lambda_{\min}(A_{EE}) = \inf_{\mathbf{x} \in \mathbb{R}^{|E|} }\frac{\mathbf{x}^TA_{EE}\mathbf{x}}{\mathbf{x}^T\mathbf{x}}.
$$
Given $\mathbf{x} \in \mathbb{R}^{|E|}$, let $\tilde{\mathbf{x}}$ denote its augmentation in $\mathbb{R}^p$ such that the entries of$\tilde{\mathbf{x}}$ indexed by $E$ coincide with that of $\mathbf{x}$, and those indexed by $E^c$ are equal to 0.  Then
$$
    \lambda_{\min}(A_{EE}) = \inf_{\mathbf{x} \in \mathbb{R}^{|E|} }\frac{\tilde{\mathbf{x}}^TA\tilde{\mathbf{x}}}{\tilde{\mathbf{x}}^T\tilde{\mathbf{x}}} \ge \inf_{\tilde{\mathbf{x}} \in \mathbb{R}^{p} }\frac{\tilde{\mathbf{x}}^TA\tilde{\mathbf{x}}}{\tilde{\mathbf{x}}^T\tilde{\mathbf{x}}} = \lambda_{\min}(A).
$$
The upper bound for $\lambda_{\max}(A_{EE})$ is proved similarly.
\end{proof} 

\begin{proof}[Proof of Proposition~\ref{prop:mi:upper:bound}]
Observe that
\begin{eqnarray*}
|||\Omega_{E^cE}(\Omega_{EE})^{-1}|||_\infty &=& \sup_{\mathbf{x} \in \mathbb{R}^{|E|}} \frac{\|\Omega_{E^cE}(\Omega_{EE})^{-1}\mathbf{x}\|_\infty}{\|\mathbf{x}\|_\infty} \\
&=& \sup_{\mathbf{y} \in \mathbb{R}^{|E|}} \frac{\|\Omega_{E^cE}\mathbf{y}\|_\infty}{\|\Omega_{EE}\mathbf{y}\|_\infty} \\
&=& \sup_{\|\mathbf{y}\|_2=1} \frac{\|\Omega_{E^cE}\mathbf{y}\|_\infty}{\|\Omega_{EE}\mathbf{y}\|_\infty} \\
&\le& \frac{\sup_{\|\mathbf{y}\|_2=1} \|\Omega_{E^cE}\mathbf{y}\|_\infty}{\inf_{\|\mathbf{y}\|_2=1} \|\Omega_{EE}\mathbf{y}\|_\infty} \\
&\le& \sqrt{|E|} \cdot \frac{\sup_{\|\mathbf{y}\|_2=1} \|\Omega_{E^cE}\mathbf{y}\|_2}{\inf_{\|\mathbf{y}\|_2=1} \|\Omega_{EE}\mathbf{y}\|_2} \\
&\le& \left( |E| \cdot \frac{\sup_{\|\mathbf{y}\|_2=1} \mathbf{y}^T\Omega_{EE^c}\Omega_{E^cE}\mathbf{y}}{\inf_{\|\mathbf{y}\|_2=1} \mathbf{y}^T\Omega_{EE}\Omega_{EE}\mathbf{y}} \right)^{1/2} \\
&\le& \left( |E| \cdot \left (\frac{\sup_{\|\mathbf{y}\|_2=1} \mathbf{y}^T(\Omega_{EE^c}\Omega_{E^cE}+\Omega_{EE}\Omega_{EE})\mathbf{y}}{\inf_{\|\mathbf{y}\|_2=1} \mathbf{y}^T\Omega_{EE}\Omega_{EE}\mathbf{y}} -1 \right)\right)^{1/2} \\
&\le& \left( |E| \cdot \left (\frac{\lambda_{\max}((\Omega^2)_{EE})}{\lambda_{\min}((\Omega_{EE})^2)} -1 \right)\right)^{1/2},
\end{eqnarray*}
where we used the inequality that for any $\mathbf{x} \in \mathbb{R}^{p}$,
$$
	\|\mathbf{x}\|_2 \le \|\mathbf{x}\|_\infty \le \sqrt{p} \|\mathbf{x}\|_2.
$$
From Lemma~\ref{lemma:2}, 
$$
    \lambda_{\max}((\Omega^2)_{EE}) \le \lambda_{\max}(\Omega^2) = \lambda^2_{\max}(\Omega)
$$
and
$$
    \lambda_{\min}((\Omega_{EE})^2) = \lambda_{\min}^2(\Omega_{EE}) \ge \lambda_{\min}^2(\Omega).
$$
Hence the inequality can be further written as
\begin{eqnarray*}
|||\Omega_{E^cE}(\Omega_{EE})^{-1}|||_\infty &\le& \left( |E| \cdot \left (\frac{\lambda_{\max}((\Omega^2)_{EE})}{\lambda_{\min}((\Omega_{EE})^2)} -1 \right)\right)^{1/2} \\
&\le& \left( |E| \cdot \left (\frac{\lambda_{\max}^2(\Omega)}{\lambda_{\min}^2(\Omega)} -1 \right)\right)^{1/2}.
\end{eqnarray*}

It remains to derive the expressions for $\lambda_{\min}^2(\Omega)$ and $\lambda_{\max}^2(\Omega)$.
From Proposition~\ref{prop:sigma:theta}, $\Sigma$ has ordered eigenvalues
$$
	0= \lambda_1\le \lambda_2 \cdots \le \lambda_d,
$$
with corresponding eigenvectors
$$
	\mathbf{u}_1=\frac{1}{\sqrt{d}}\mathbf1, \mathbf{u}_2 \ldots,\mathbf{u}_d.
$$  
Then $\Sigma^* = \Sigma + M \cdot \mathbf{u}_1\mathbf{u}_1^T$ is positive definite and has the same set of eigenvectors with corresponding eigenvalues
$$
	\tilde\lambda_1=M, \lambda_2,\ldots, \lambda_d.
$$

Now consider the eigendecomposition of $\Omega = \Sigma^* \otimes \Sigma^*$. 
For any $i,j$, we have
$$
	\Sigma^* \mathbf{u}_i\mathbf{u}_j^T \Sigma^* = \lambda_i\lambda_j \cdot \mathbf{u}_i\mathbf{u}_j^T,
$$
which implies that
$$
	\Omega \overline{\mathbf{u}_i\mathbf{u}_j^T} = \Sigma^* \otimes \Sigma^* \overline{\mathbf{u}_i\mathbf{u}_j^T} = \overline{\Sigma^* \mathbf{u}_i\mathbf{u}_j^T \Sigma^*} = \lambda_i\lambda_j  \overline{\mathbf{u}_i\mathbf{u}_j^T}.
$$
Hence each $\overline{\mathbf{u}_i\mathbf{u}_j^T}$ is an eigenvector for $\Omega$ with corresponding eigenvalue $\lambda_i\lambda_j$.
It remains to show that $\{\overline{\mathbf{u}_i\mathbf{u}_j^T}\}_{i,j}$ forms an orthonormal basis.  For any $(i,j),(i',j')$,
$$
 \overline{\mathbf{u}_i\mathbf{u}_j^T}^T  \overline{\mathbf{u}_{i'}\mathbf{u}_{j'}^T} =
  tr(\mathbf{u}_i\mathbf{u}_j^T\mathbf{u}_{j'}\mathbf{u}_{i'}^T)\\
=
 \begin{cases}
 0, & j\ne j' \\
 tr(\mathbf{u}_i\mathbf{u}_{i'}^T), & j=j'
 \end{cases}
 =
 \begin{cases}
 0, & j\ne j' \\
 0,& j=j', i\ne i' \\
1& j=j',i=i'
 \end{cases}.
$$
Therefore $\Omega$ has the set of eigenvalues
$$
	\{\gamma_{ij}\}_{i,j}:=\{\lambda_i\lambda_j\}_{i,j},
$$
with eigenvectors
$$
	\left\{\mathbf{v}_{ij}\right\}_{i,j}:=\{\overline{\mathbf{u}_i\mathbf{u}_j^T}\}_{i,j},
$$
In particular, the smallest and the largest eigenvalues correspond to
$$
    \lambda_{\min}(\Omega) = \min_{i,j}\{\lambda_i\lambda_j\} = \min\{M^2,\lambda_2^2\}
$$
and
$$
    \lambda_{\max}(\Omega) = \max_{i,j}\{\lambda_i\lambda_j\} = \max\{M^2,\lambda_d^2\}.
$$
\end{proof}

\section{Block coordinate descent algorithm for the extreme graphical lasso} \label{app:algo}


\subsection{Problem specification}

In this appendix we present the algorithm for solving for $\Theta$ given the input of $S$ and $c>0$ from the following problem.
\begin{equation} \label{eqapp:eglasso}
	\min_{\Theta \succ 0} -\log|\Theta| + tr(S\Theta) + \lambda \sum_{i\ne j} |\theta_{ij}-c|.
\end{equation}
We slightly deviate from the notation from the rest of the paper.  The $\Theta$ here corresponds to the $\Theta^*$ in the main text of the manuscript.  This is to simply the notation in the following algebraic representations.

The classical graphical lasso requires $c=0$ instead of $c>0$.  Similar to $c=0$, for $c>0$, \eqref{eqapp:eglasso} is also a convex optimization.  Therefore a unique solution exists and satisfies the KKT condition:
\begin{equation} \label{eqapp:kkt}
	-\Theta^{-1} + S + \lambda \Gamma=\mathbf0,
\end{equation}
where $\Gamma := \text{sign}(\Theta)$ is the sub-gradient of $\sum_{i\ne j} |\theta_{ij}-c|$ satisfying
$$
	\gamma_{ij}  =
	\begin{cases}
	\text{sign}(\theta_{ij}-c), & \theta_{ij} \ne c;\\
	\in [-1,1], & \theta_{ij}=c.
	\end{cases}
$$

The algorithm presented in this appendix is adapted from the P-GLASSO algorithm proposed by \cite{mazumder2012graphical}.  While the most widely used algorithm for solving the classical graphical lasso is the GLASSO algorithm \citep{fri2008}, based on block coordinate descent, it cannot be directly adapted for the case of $c>0$.  The P-GLASSO algorithm is also based on block coordinate descent but solves for a different target function at each iteration.  It has a similar order of computation as the GLASSO and is designed to avoid the computation incurred by inverting large matrices at each iteration.


\subsection{Block matrix inversion}

In this algorithm we make use of the following result for the inversion of block matrices.  

Let $W$ be the inverse of $\Theta$ such that $W \cdot \Theta = I$, where $$
	\Theta = 
	\begin{pmatrix}
	\Theta_{11} & \boldsymbol\theta_{12} \\
	 \boldsymbol\theta_{12}^T & \theta_{22}
	\end{pmatrix}.
$$
Then
\begin{eqnarray}
W=\left(
	\begin{array}{cc}
	W_{11} & \mathbf{w}_{12} \\
	\mathbf{w}_{12}^T & w_{22}
	\end{array}
	\right)
	&=&
\left(
	\begin{array}{cc}
	(\Theta_{11} -  \theta_{22}^{-1}\boldsymbol\theta_{12}\boldsymbol\theta_{12}^T)^{-1} & -\theta_{22}^{-1}W_{11}\boldsymbol\theta_{12} \\
	\cdot & \theta_{22}^{-1} + \theta_{22}^{-2}\boldsymbol\theta_{12}^TW_{11}\boldsymbol\theta_{12}
	\end{array}
	\right) \label{eqapp:inverse1} \\
	&=&
\left(
	\begin{array}{cc}
	\Theta_{11}^{-1} +\frac{\Theta_{11} ^{-1}\boldsymbol\theta_{12}\boldsymbol\theta_{12}^T\Theta_{11}^{-1} }{\theta_{22} - \boldsymbol\theta_{12}^T\Theta_{11} ^{-1} \boldsymbol\theta_{12}} 
		& -\frac{\Theta_{11} ^{-1}\boldsymbol\theta_{12}}{\theta_{22} - \boldsymbol\theta_{12}^T\Theta_{11} ^{-1} \boldsymbol\theta_{12}} \\
	\cdot 
		& \frac{1}{\theta_{22} - \boldsymbol\theta_{12}^T\Theta_{11} ^{-1} \boldsymbol\theta_{12}},
	\end{array}
	\right) \label{eqapp:inverse2}
\end{eqnarray}
where $\cdot$ denotes the mirroring of elements in the upper triangle. The derivation can be found in Section~\ref{app:algo:proofs}.  The same formula can be applied for the opposite representation of $\Theta$ using $W$.


\subsection{P-GLASSO algorithm}

The algorithm operates by updating a specific row/column of $\Theta$ ($\Theta$ is symmetric) in each iteration while keeping the rest of the row/column fixed.  The operation is repeated while iterating through the rows/columns until convergence.  Write $\Theta$ as
$$
	\Theta = 
	\begin{pmatrix}
	\Theta_{11} & \boldsymbol\theta_{12} \\
	 \boldsymbol\theta_{12}^T & \theta_{22}
	\end{pmatrix}.
$$
Consider the problem of estimating $\boldsymbol\theta_{12}$ and $\theta_{22}$ while keeping $\Theta_{11}$ fixed.

Denote $W$ as the inverse of $\Theta$.  Then the KKT condition can be re-written as
\begin{equation} \label{eqapp:kkt:2}
	-W + S + \lambda \Gamma=\mathbf0.
\end{equation}
First, \eqref{eqapp:kkt:2} provides direct solutions to the diagonal entries of $W$:
$$
	w_{ii} = s_{ii} + \lambda \gamma_{ii} = s_{ii}.
$$
The last column of \eqref{eqapp:kkt:2} gives
\begin{equation} \label{eqapp:kkt:lower}
	-\mathbf{w}_{12} + \mathbf{s}_{12} + \lambda \boldsymbol\gamma_{12} = \mathbf0.
\end{equation}
The lower left entry of the right hand side of \eqref{eqapp:inverse1} gives
$$
	\Theta_{11}\mathbf{w}_{12} + \boldsymbol\theta_{12} \cdot w_{22} = \mathbf0,
$$
which allows us to write $\mathbf{w}_{12}$ as
$$
	\mathbf{w}_{12} = - \Theta_{11}^{-1} \boldsymbol\theta_{12} \cdot w_{22}.
$$
Plugging into \eqref{eqapp:kkt:lower}, we have
$$
	\Theta_{11}^{-1}\cdot  \boldsymbol\theta_{12} w_{22} + \mathbf{s}_{12} + \lambda \cdot \text{sign}(\boldsymbol\theta_{12} - c \mathbf1) = \mathbf0.
$$
Replacing $w_{22}$ with the known solution $w^*_{22} := s_{22} + \lambda$, we have
$$
	\Theta_{11}^{-1}\cdot  \boldsymbol\theta_{12} w^*_{22} + \mathbf{s}_{12} + \lambda \cdot \text{sign}(\boldsymbol\theta_{12} - c \mathbf1) = \mathbf0.
$$
Here we reserve the notation $w_{22}$ for the lower diagonal entry of the current $W$ for reasons to be stated later.
Denote
$$
	\boldsymbol\alpha:= \left(\boldsymbol\theta_{12}- c \mathbf1\right) w^*_{22}.
$$
Then
$$
	\Theta_{11}^{-1}\cdot  \boldsymbol\alpha + \Theta_{11}^{-1}c \mathbf1 \cdot w^*_{22} + \mathbf{s}_{12} + \lambda \cdot \text{sign}( \boldsymbol\alpha) = \mathbf0.
$$
This equation is the KKT condition for the lasso problem
$$
	\min_{\boldsymbol\alpha} \left\{ \frac{1}{2} \boldsymbol\alpha\Theta_{11}^{-1}\boldsymbol\alpha + \left(\Theta_{11}^{-1}c \mathbf1 \cdot w^*_{22}  + \mathbf{s}_{12}\right)^T \boldsymbol\alpha + \lambda \|\boldsymbol\alpha\|_1\right\},
$$
which can be solved using a standard lasso algorithm.
Once $\boldsymbol\alpha$ is solved, $\boldsymbol\theta_{12}$ can be obtained from
$$
	\boldsymbol\theta_{12} = \frac{\boldsymbol\alpha}{w^*_{22}} + c \mathbf1.
$$
Then $\theta_{22}$ can be updated by
\begin{equation} \label{eqapp:update:22}
	\theta_{22} = \frac{1}{w_{22}^*} + \boldsymbol\theta_{12}^T \Theta_{11}^{-1} \boldsymbol\theta_{12},
\end{equation}
which follows from the matrix operation of inverses, see \eqref{eqapp:inverse2}.  

In principle, we can repeat the above procedure to iteratively update each row/column of $\Theta$.  The estimated $\Theta$ will guarantee to converge from the convergence of block gradient descent, see for example Proposition~2.7.1 of \cite{bertsekas1997nonlinear}. However, there is one extra consideration.  We would like to avoid inverting large matrices, specifically in calculating $\Theta_{11}^{-1}$.  Therefore we require an extra step for storing the updates in each iteration that proceeds as follows.

We propose to store all our updates from each iteration into the matrix $W$ and pass it on to the next iteration:
\begin{itemize}
\item
	At the beginning of an iteration, given $W$, we can calculate $\Theta_{11}^{-1}$ from
	$$
		\Theta_{11}^{-1} = W_{11} - \mathbf{w}_{12}\mathbf{w}_{12}^T/w_{22}.
	$$
	This follows from \eqref{eqapp:inverse2}.  Here $w_{22}$, the current lower diagonal entry of $W$, is used instead of $w^*_{22}$, in order to ensure the positive definiteness of $\Theta_{11}^{-1}$.
\item
	At the end of an iteration, we can calculate the updated $W$ from $\Theta_{11}^{-1}$, $\boldsymbol\theta_{12}$ and $\theta_{22}$ using the formula \eqref{eqapp:inverse2}.  The updated $W$ is then carried on to the next iteration.
\end{itemize}
In both operations, only matrix multiplication is required, minimizing the computation incurred.

In the following we show that for each iteration, if the input $W$ is positive definite, then the updated $W$ will also be positive definite, hence ensuring that $\Theta$ is positive definite throughout the iterations.

Assume that input $W$ is positive definite.  Then $\Theta_{11}$, which is fixed in the current iteration, is also positive definite.  From the Schur complement lemma, the necessary and sufficient condition for the updated $\Theta$ to be positive definite is
$$
	\theta_{22} - \boldsymbol\theta_{12}^T \Theta_{11}^{-1} \boldsymbol\theta_{12} > 0.
$$
This is readily satisfied, as by the updating rule \eqref{eqapp:update:22},
$$
	\theta_{22} - \boldsymbol\theta_{12}^T \Theta_{11}^{-1} \boldsymbol\theta_{12} = \frac{1}{w^*_{22}} >0.
$$

%


\subsection{Pseudo code: P-GLASSO for extremes}

\begin{itemize}
\item
Input: $S$, $\lambda$ and $c$. 
\item
Initiate $W=S$.
(Or Initiate $W = diag(S)$ and $\Theta = W^{-1}$.)
\item
Iterate until convergence.  Input $W$.  Output updated $W$.
\begin{itemize}
\item
	Permute the rows/columns of $W$ such that the target row/column is at the end of the matrix (i.e.~$\mathbf{w}_{12}$).
\item
	Compute $\Theta_{11}^{-1}$ from input $W$ via
	$$
		\Theta_{11}^{-1} = W_{11} - \mathbf{w}_{12}\mathbf{w}_{12}^T/w_{22}.
	$$
\item
	Let $\boldsymbol\alpha:= (\boldsymbol\theta_{12} - c\mathbf1) w^*_{22}$, where $w^*_{22}:=s_{22}$.  Solve for $\boldsymbol\alpha$ via
	$$
		\boldsymbol\alpha =   {\arg\min}_{\boldsymbol\alpha} \left\{\frac12 \boldsymbol \alpha^T\Theta_{11}^{-1} \boldsymbol\alpha + (\mathbf{s}_{12} +cw^*_{22} \Theta_{11}^{-1} \mathbf1)^T \boldsymbol\alpha + \lambda\|\boldsymbol\alpha\|_1\right\}.
	$$
	This solution can be written in a lasso regression form
	$$
		\boldsymbol\alpha =  {\arg\min}_{\boldsymbol\alpha} \left\{\frac12 \| A \boldsymbol\alpha - \mathbf{b} \|^2_2 + \lambda\|\boldsymbol\alpha\|_1\right\},
	$$		
	where $A$ and $\mathbf{b}$ satisfy
	$$
		A^TA = \Theta_{11}^{-1},
	$$
	(here $\Theta_{11}^{-1}$ is positive definite and symmetric, hence $A = A^T$ is symmetric as well,)
	and
	$$
		A\mathbf{b} = -(\mathbf{s}_{12} + cw^*_{22}  \Theta_{11}^{-1}\mathbf1).
	$$
	Let the eigendecomposition of $\Theta_{11}^{-1}$ be
	$$
		\Theta_{11}^{-1} = Q \Lambda Q^T,
	$$
	where the columns of $Q$ consists of the eigenvectors of $\Theta_{11}^{-1}$ and $\Lambda$ is the diagonal matrix with whose diagonal values correspond to the eigenvalues of $\Theta_{11}^{-1}$.  Then $A$ and $A^{-1}$ can be calculated from
	$$
		A =  Q \Lambda^{\frac12} Q^T,
	$$
	and
	$$
		A^{-1} =  Q \Lambda^{-\frac12} Q^T.
	$$
\item
	Update
	$$
		\boldsymbol\theta_{12} = \frac{1}{w^*_{22}}\boldsymbol\alpha + c\mathbf1.
	$$
\item
	Update
	$$
		\theta_{22} = \frac{1}{w^*_{22}} + \boldsymbol\theta_{12}^T \Theta_{11}^{-1} \boldsymbol\theta_{12}.
	$$
\item
	Update
	\begin{eqnarray*}
		W &=&
	\begin{pmatrix}
	\Theta_{11}^{-1} +\frac{\Theta_{11} ^{-1}\boldsymbol\theta_{12}\boldsymbol\theta_{12}^T\Theta_{11}^{-1} }{\theta_{22} - \boldsymbol\theta_{12}^T\Theta_{11} ^{-1} \boldsymbol\theta_{12}} 
		& -\frac{\Theta_{11} ^{-1}\boldsymbol\theta_{12}}{\theta_{22} - \boldsymbol\theta_{12}^T\Theta_{11} ^{-1} \boldsymbol\theta_{12}} \\
	\cdot 
		& \frac{1}{\theta_{22} - \boldsymbol\theta_{12}^T\Theta_{11} ^{-1} \boldsymbol\theta_{12}}
	\end{pmatrix}\\
	&=& 
	\begin{pmatrix}
	\Theta_{11}^{-1} +w^*_{22} \cdot \Theta_{11} ^{-1}\boldsymbol\theta_{12}\boldsymbol\theta_{12}^T\Theta_{11}^{-1} 
		& -w^*_{22} \cdot \Theta_{11} ^{-1}\boldsymbol\theta_{12}\\
	\cdot 
		&w^*_{22} 
	\end{pmatrix}.
	\end{eqnarray*}
\item
	Arrange the rows and columns of $W$ back to the original order.
\end{itemize}
\item
Calculate the final estimate $\Theta = W^{-1}$.
\end{itemize}


\subsection{Proofs of \eqref{eqapp:inverse1} and \eqref{eqapp:inverse2}} \label{app:algo:proofs}

We now prove the equations \eqref{eqapp:inverse1} and \eqref{eqapp:inverse2}.  From the block definition of $W$ and $\Theta$,
\begin{eqnarray}
	W_{11} \Theta_{11} + \mathbf{w}_{12}\boldsymbol\theta_{12}^T &=& I \label{eqapp:inveq1}\\
	W_{11} \boldsymbol\theta_{12} + \theta_{22} \cdot \mathbf{w}_{12} &=& \mathbf{0} \label{eqapp:inveq2}\\
	\mathbf{w}_{12}^T \Theta_{11} + w_{22} \cdot \boldsymbol\theta_{12}^T &=& \mathbf{0}^T \label{eqapp:inveq3}\\
	\mathbf{w}_{12}^T \boldsymbol\theta_{12} + w_{22}\cdot \theta_{22} &=& 1.\label{eqapp:inveq4}
\end{eqnarray}
From \eqref{eqapp:inveq2},
$$
	 \mathbf{w}_{12} = - \theta_{22}^{-1}W_{11} \boldsymbol\theta_{12}.
$$
Plugging into \eqref{eqapp:inveq4}, 
$$
	w_{22}  =  \theta_{22}^{-1} (1 - \mathbf{w}_{12}^T \boldsymbol\theta_{12}) = \theta_{22}^{-1} + \theta_{22}^{-2}\boldsymbol\theta_{12}^TW_{11}\boldsymbol\theta_{12}.
$$
From \eqref{eqapp:inveq1},
$$
	W_{11}^{-1} = \Theta_{11} + W_{11}^{-1} \mathbf{w}_{12}\boldsymbol\theta_{12}^T
	= \Theta_{11} - W_{11}^{-1} \theta_{22}^{-1}W_{11} \boldsymbol\theta_{12}\boldsymbol\theta_{12}^Td = \Theta_{11} - \theta_{22}^{-1}\boldsymbol\theta_{12}\boldsymbol\theta_{12}^T 
$$
This proves \eqref{eqapp:inverse1}.  

Now we prove \eqref{eqapp:inverse2}.  From \eqref{eqapp:inveq1} and \eqref{eqapp:inveq3},
\begin{eqnarray}
	W_{11} &=& \Theta_{11}^{-1} - \mathbf{w}_{12}\boldsymbol\theta_{12}^T\Theta_{11}^{-1} \label{eqapp:eq1} \\
	 \mathbf{w}_{12} &=& -w_{22}\Theta_{11} ^{-1}\boldsymbol\theta_{12} \label{eqapp:eq2} 
\end{eqnarray}
Plugging \eqref{eqapp:eq2} into \eqref{eqapp:inveq4},
$$
	 w_{22}\theta_{22}  -w_{22}\boldsymbol\theta_{12}^T\Theta_{11} ^{-1} \boldsymbol\theta_{12} = 1 
$$
and consequently,
\begin{equation} \label{eqapp:w22}
	w_{22} = \frac{1}{\theta_{22} - \boldsymbol\theta_{12}^T\Theta_{11} ^{-1} \boldsymbol\theta_{12}}.
\end{equation}
Plugging \eqref{eqapp:w22} into \eqref{eqapp:eq2}, 
$$
	\mathbf{w}_{12} = -\frac{\Theta_{11} ^{-1}\boldsymbol\theta_{12}}{\theta_{22} - \boldsymbol\theta_{12}^T\Theta_{11} ^{-1} \boldsymbol\theta_{12}}.
$$
Plugging \eqref{eqapp:w22} into \eqref{eqapp:eq1},
$$
	W_{11} = \Theta_{11}^{-1} +\frac{\Theta_{11} ^{-1}\boldsymbol\theta_{12}\boldsymbol\theta_{12}^T\Theta_{11}^{-1} }{\theta_{22} - \boldsymbol\theta_{12}^T\Theta_{11} ^{-1} \boldsymbol\theta_{12}}.
$$
This proves \eqref{eqapp:inverse2}.



\end{document}